\documentclass{article}

\usepackage{natbib}
\setcitestyle{authoryear}

\usepackage[english]{babel}
\usepackage[nottoc]{tocbibind}

\usepackage{xcolor}
  
\usepackage{amsmath,amsthm,amssymb,graphicx,comment,dsfont,fullpage, mathtools, tabu, multirow, url, indentfirst}
\usepackage{wrapfig}
\usepackage{algorithm2e, pdflscape, hyperref, makecell,sidecap}
\usepackage{authblk}
\RestyleAlgo{ruled}

\usepackage[shortlabels]{enumitem}

\usepackage{floatrow}
\usepackage[font=bf]{subfig}

\usepackage{amsmath}
\DeclareMathOperator*{\argmax}{arg\,max}

\makeatletter
\providecommand\phantomcaption{\caption@refstepcounter\@captype}
\makeatother

\AtBeginDocument{}

\newtheorem{theorem}{Theorem}

\newtheorem{proposition}[theorem]{Proposition}

\hypersetup{linkcolor=blue, citecolor=blue, colorlinks=true}

\title{Multiple merger coalescent inference of effective population size}
\author{Julie Zhang, Julia A. Palacios}
\date{}

\begin{document}

\maketitle

\begin{abstract}   
Variation in a sample of molecular sequence data informs about the past evolutionary history of the sample's population. Traditionally, Bayesian modeling coupled with the standard coalescent, is used to infer the sample's bifurcating genealogy and demographic and evolutionary parameters such as effective population size, and mutation rates. However, there are many situations where binary coalescent models do not accurately reflect the true underlying ancestral processes. Here, we propose a Bayesian nonparametric method for inferring effective population size trajectories from a multifurcating genealogy under the $\Lambda-$coalescent. In particular, we jointly estimate the effective population size and model parameters for the Beta-coalescent model, a special type of $\Lambda-$coalescent. Finally, we test our methods on simulations and apply them to study various viral dynamics as well as Japanese sardine population size changes over time. The code and vignettes can be found in the \texttt{phylodyn} package. 

\end{abstract}

\section{Introduction} \label{sec:intro}
In phylodynamic inference, variation in a sample of molecular sequence data is used to learn about the past ancestral history of a sample, usually represented by a bifurcating genealogy. Kingman's coalescent is then typically used as a prior model on the genealogy, parameterized in terms of a parameter of interest called the effective population size $N_e(t)$ \citep{Kingman1982, slatkin2001simulating}. It has been shown that Kingman's coalescent is a good approximation to the distribution of the sample's ancestry under several underlying population dynamics \citep{Wakeley2009a}. However, Kingman's coalescent assumes the population variance in the number of offspring is sufficiently small such that at most two lineages merge at a time. Such bifurcating tree models may not fit all population dynamics of interest. This can occur in the study of infectious disease dynamics in the presence of superspreader events. Examples include tuberculosis \citep{menardo2021multiple} and wild type polio \citep{li2017quantifying}. The same situation arises in highly fecund species populations, where certain individuals are capable of reproducing offspring on the order of the population size; this is known as reproductive skew or sweepstakes reproduction \citep{eldon2018evolution,eldon2020evolutionary}. Some marine species undergo sweepstakes reproduction, such as the Japanese sardine \citep{niwa2016reproductive}, Pacific oysters \citep{sargsyan2008coalescent}, and Korean seaweed \citep{byeon2019origin}. Understanding the phylodynamics of these populations may have impacts on the economy and the environment. Multifurcating trees can also arise in populations undergoing strong positive selection \citep{der2012dynamics, eldon2023sweepstakes,arnason2023sweepstakes}, and large sample sizes from small populations \citep{wakeley2003gene}. Here we assume neutral evolution without recombination. In addition, we only consider multifurcating trees without simultaneous mergers, a situation that would arise in diploid populations with sweepstakes reproduction \citep{mohle2003coalescent}. \\

The $\Lambda$-coalescent is a multiple merger coalescent model (MMC) proposed as a generalization to binary coalescent processes \citep{pitman1999coalescents, sagitov1999general}. The $\Lambda$-coalescent includes Kingman's coalescent, the Beta-coalescent \citep{berestycki2007beta} and the Psi-coalescent \citep{eldon2006coalescent} as special cases. In particular, the Beta-coalescent is a one-parameter model that was shown to be the ancestral limit of a sample obtained from the Cannings population model with a heavy tail offspring distribution and constant population size \citep{schweinsberg2003coalescent,berestycki2009recent}. \cite{hoscheit2019multifurcating} extended the Beta-coalescent model to allow variable effective population sizes and heterochronous sampling, i.e. when tips have different dates. The effective population size under Kingman's coalescent can be interpreted as the population size under the Wright-Fisher model that has the same genetic drift as the population under study \citep{Wakeley2009a}. Under the Beta-coalescent, a similar interpretation holds for a population that evolves under the Cannings population model with a specific offspring distribution. See \cite{eldon2020evolutionary} for a recent review of MMC models and \cite{korfmann2024simultaneous} for modeling recombining genealogies with MMC.\\ 

Most applications of the $\Lambda$-coalescent aimed to distinguish between population growth, such as exponential growth, and multiple merger genealogy from summaries of molecular data including the site frequency spectrum (SFS) \citep{eldon2015can,koskela2018multi,matuszewski2018coalescent,koskela2019robust,freund2021impact}. A recently proposed method estimates the base measure $\Lambda$ that describes the coalescent rates from the SFS \citep{pina2023estimating}. However all these methods assume constant effective population size or a fixed growth model and the infinite sites mutation model. Instead, a method for inferring effective population sizes under the Beta-coalescent from a given multiple merger genealogy is proposed in \cite{hoscheit2019multifurcating} . Although this method ignores genealogical uncertainty, it opens the door to applications in which an estimated genealogy is obtained under any mutation model (for example via maximum likelihood estimation). The authors propose an extension of the classic skyline plot \citep{ho2011skyline} for inferring the effective population size, however there is no joint inference of effective population size and the characteristic measure that describes coalescent rates. \\

In this manuscript, we propose a method for joint inference of the effective population size $N_{e}(t)$ and the characteristic coalescent measure under the Beta-coalescent. We expand upon the work of \cite{pal12, lan2015efficient}, and develop a Bayesian nonparametric phylodynamic approach that relies on Gaussian Markov random field priors on $N_{e}(t)$. We evaluate the performance of proposed methods on simulations and apply them to analyze two infectious diseases. Finally, we re-analyze the reproductive skew hypothesis in Japanese sardine populations \citep{niwa2016reproductive} using our methods. Those who are interested in examining case studies can focus on Sections \ref{sec:results} and \ref{sec:applications}.

\section{Background on $\Lambda$-coalescent} \label{sec:background}

The standard Kingman's $n$-coalescent is a backward-in-time Markov jump chain on binary partitions of $[n]=\{1,\ldots,n\}$ whose full realization is a genealogy of $n$ individuals. The process starts with $n$ singleton lineages at time 0. At each step, two lineages are chosen uniformly at random to coalesce, continuing until there is a single lineage at the root \citep{Kingman1982}. The coalescent holding times are exponentially distributed with rate $\binom{A(t)}{2}$, when there are $A(t)$ ancestral lineages at time $t$. Kingman's coalescent has also been extended to incorporate variable population size, also termed the effective population size $N_e(t)$ \citep{slatkin2001simulating}, and samples collected at different times (heterochronous sampling) \citep{Felsenstein1999}. Here, $N_e(t)$ is interpreted as the size of an ideal population that exhibits the same level of genetic drift under a standard Wright-Fisher population model \citep{Wakeley2009b}, and can be thought of as a relative measure of genetic diversity over time. The effective population size affects the coalescent times: larger $N_e(t)$ implies smaller coalescent rates and longer time until coalescence. \\

The $\Lambda$-coalescent is a generalization of the standard coalescent that allows for multiple mergers \citep{sagitov1999general,pitman1999coalescents}. The process is a Markov jump chain on partitions of $[n]=\{1,\ldots,n\}$ whose full realization is a multiple merger genealogy of $n$ individuals. The process starts with $n$ singleton lineages at time 0. At each step, two or more lineages are chosen to coalesce, continuing until there is a single lineage at the root (Figure \ref{fig:multif_tree_ex}). Here, $k\geq 2$ lineages merge at rate $\lambda_{A(t),k}$, where $A(t)$ is the number of extant lineages at time $t$. The coalescent holding times are exponentially distributed with rate $\lambda_{A(t)} = \sum_{k=2}^{A(t)} \binom{A(t)}{k} \lambda_{A(t), k}$, also called the total coalescent rate. Conditional on a coalescent event at time $t$ when there are $A(t)$ lineages, the distribution of the block (or multiple merger) size is %we sample the number of coalescing lineages $m$ with from a discrete distribution on $\{2,..., A(t)\}$ where 
\begin{equation} \label{eq:block_lik_rates}
    \mathbb{P}(X=k) = \frac{\binom{A(t)}{k}  \lambda_{A(t), k}}{ \sum_{i=2}^{A(t)} \; \; \binom{A(t)}{i} \lambda_{A(t),i}}, 2 \leq k \leq A(t).
\end{equation}
The $k$ lineages in the block are then chosen uniformly at random among the $A(t)$ lineages. 
The rate $\lambda_{b,k}$ at which a specific $k$-block of lineages merges when there are $b$ lineages is defined to be
\begin{equation} \label{eq:lambda_merge_rates}
    \lambda_{b,k} = \int_0^1 x^{k-2} (1-x)^{b-k} \Lambda(dx),
\end{equation}
where $\Lambda$ is a measure on $[0,1]$. \\

\cite{pitman1999coalescents} show that any simple Markovian MMC process that satisfies: exchangeability of the lineages and consistency of the merger rates defined by $\lambda_{b,k} = \lambda_{b+1, k} + \lambda_{b+1, k+1}$ must be the $\Lambda$-coalescent. The intuition behind the consistency condition is as follows. Suppose we have a specific $k$-block of lineages out of $b$ lineages. If there are $b+1$ total lineages, then the extra lineage can coalesce with the $k$-block at rate $\lambda_{b+1, k+1}$  or the extra lineage does not coalesce with the $k$-block at rate $\lambda_{b+1, k}$, and so the additive property must hold. The $\Lambda$-coalescent can also be constructed as a Poisson point process on $(0,1]\times (0,\infty)$ with intensity measure $\Lambda(dx)x^{-2} \otimes dt$; see \cite{berestycki2009recent} for a detailed theoretical construction. We can interpret points $(x_i, t_i)$ from this Poisson point process as flipping a coin with probability $x_i$ of heads for each block and merging all blocks that are heads at time $t_i$.  \\

Some special cases of the $\Lambda$-coalescent include: 
\begin{itemize}
    \item $\Lambda= \delta_0$, i.e. the point mass at 0: In this case, $\lambda_{b,k}=1$ for $k=2$ and corresponds to \textit{Kingman's coalescent.} 
    \item $\Lambda = \delta_1$, i.e. the point mass at 1: In this case, then nothing happens for an exponential amount of time with mean 1, at which point all lineages coalesce and we get the \textit{star-shaped coalescent.}
    \item $\Lambda =$ Beta$(2-\alpha,\alpha)$ for $0<\alpha<2$: In this case, $\lambda_{b,k} = \text{B}(k-\alpha, \alpha+b-k) / \text{B}(2-\alpha,\alpha)$ where $\text{B}(2-\alpha,\alpha) = \frac{\Gamma(2-\alpha) \Gamma(\alpha)}{\Gamma(2)}$ and this is called the \textit{Beta-coalescent.} A special case when $\alpha=1$ is known as the U-coalescent or the \textit{Bolthausen-Sznitman coalescent} and \[ \lambda_{b,k} = \frac{(k-2)! ( b-k)!}{(b-1)!} = \left [ (b-1) \binom{b-2}{k-2}  \right ] ^{-1} \]
\end{itemize}

The definition of the base measure $\Lambda$, and in particular the value of $\alpha$ in the Beta-coalescent, affects both the shape of the tree topology and the distribution of the coalescent times. To see this, in Figure~\ref{fig:sim-mean-block-size}(a) we plot the average block size for different values of $\alpha$ and number of tips $n$. We simulate 500 trees for each value of $\alpha$ and $n$, and then compute the average block size across all trees and multiple merger events. We see the average block size decreases to 2, the binary tree, as $\alpha \to 2$. The average block size also decreases with $n$, since trees with more tips are expected to have larger block sizes. In Figure~\ref{fig:sim-mean-block-size}(b) we show the total coalescent rate for different values of $\alpha$ and lineages $b$, divided by $\binom{b}{2}$, the total rate under Kingman's coalescent. The coalescent rate increases exponentially with $\alpha$. The smaller the value of $\alpha$, the smaller the coalescent rate resulting in genealogies with longer branch lengths than with Kingman's coalescent.

\begin{figure}[h]
\setcounter{subfigure}{0}
    \centering
    \sidesubfloat[]{\includegraphics[width=0.48\linewidth]{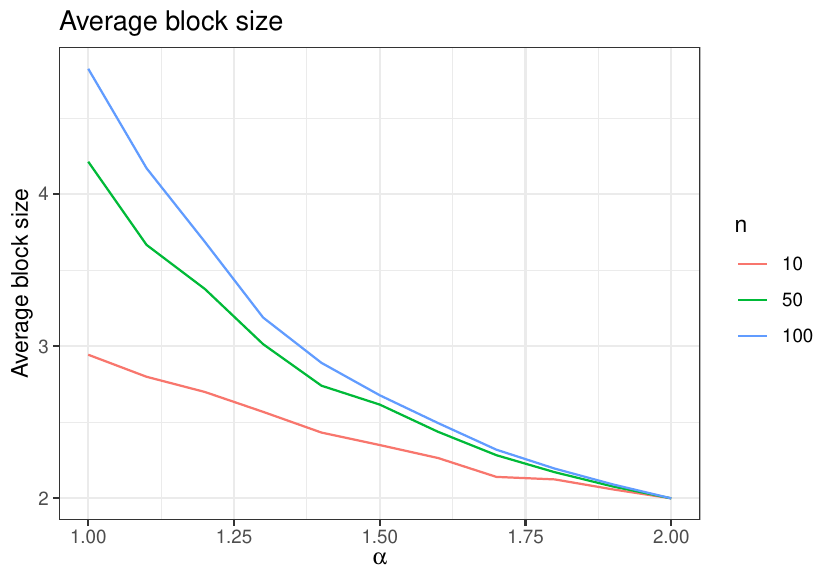}}
    \sidesubfloat[]{\includegraphics[width=0.48\linewidth]{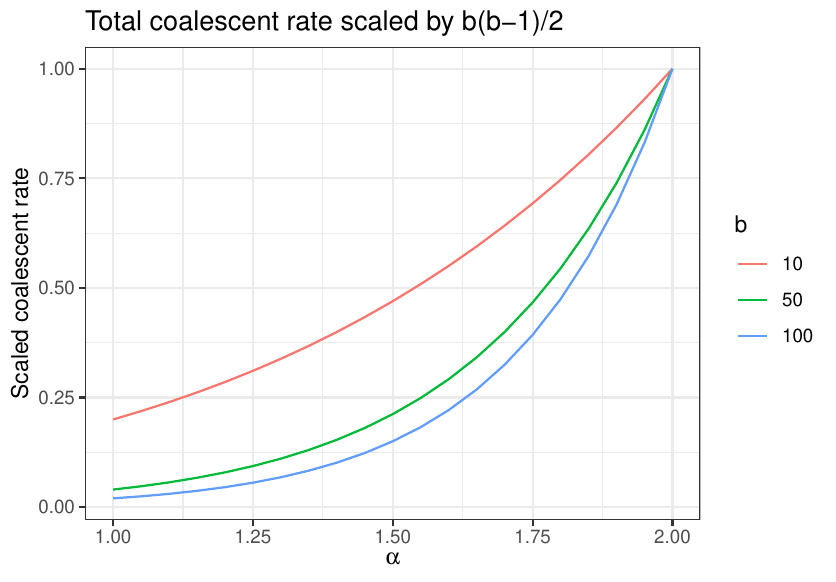}}
    \caption{\textbf{The effect of $\alpha$ in tree topology and coalescent times.} \textbf{(a):} Average block size for different values of $\alpha$ and number of tips. \textbf{(b):} Total coalescent rates $\lambda_b= \sum_{k=2}^b \binom{b}{k} \lambda_{b,k}$ when there are $b$ lineages for different values of $\alpha$, scaled by $\binom{b}{2}$, the rate under Kingman's coalescent.} 
    \label{fig:sim-mean-block-size}
\end{figure}

\newpage 
\subsection{Heterochronous $\Lambda$-coalescent with variable $N_{e}(t)$}
The $\Lambda$-coalescent was recently extended to accommodate heterochronous sampling and variable effective population size \citep{hoscheit2019multifurcating}. It is important to account for heterochronous sampling, especially for rapidly evolving organisms: the coalescent time distribution will be restricted since samples cannot coalesce before they have been sampled. Here, we assume that a rooted and timed multifurcating genealogy of $n$ haploid samples is available to us, for example, the estimated genealogy obtained via maximum likelihood from a set of $n$ observed molecular sequences. We further assume that sequences, at the tips of the genealogy, are collected at $L$ different sampling times and that the sampling process is independent of the underlying population process. Let $\boldsymbol{n}=(n_\ell)_{\ell=1:L}$ denote the number of samples collected at respective sampling times $\boldsymbol{s}=(s_\ell)_{\ell=1:L}$, with $s_1=0$, $s_{j-1} < s_{j}$ for $j=2,\dots, L$, and $n=\sum_{j=1}^{L}n_{j}$ is the total number of samples. In the genealogy, tuples of lineages merge backward in time into a common ancestor at coalescent times denoted by $\boldsymbol{t}=(t_{1},\ldots, t_{K})$, where $K$ is the step at which the most recent common ancestor (MRCA) is reached (see Figure \ref{fig:multif_tree_ex}). \\

\begin{figure}[h]
    \centering
    \includegraphics[scale=0.75]{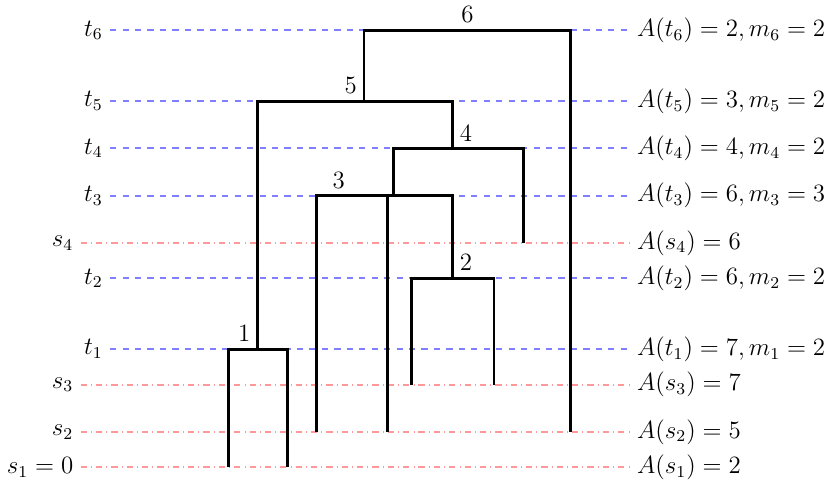}
    \caption{A multifurcating tree with 8 total lineages and 6 coalescent events labeled with the coalescent times $\boldsymbol{t}$, sampling times $\boldsymbol{s}$, block sizes $\boldsymbol{m}$, and extant lineages $A(t)$. Here there are $L=3$ sampling times.}
    \label{fig:multif_tree_ex}
\end{figure}

Let $\boldsymbol{m}=(m_{1},\ldots,m_{k})$ denote the corresponding vector of block sizes at each coalescent time. The rate at which pairs of lineages coalesce depends on the number of lineages, the $\Lambda$ measure, and the effective population size $(N_e(t))_{t\geq 0}:=N_{e}(t)$. Finally, let the number of extant lineages at time $t$ be given by $A(t)=\sum^{L}_{i=1} n_i \mathds{1}(s_{i}<t) - \sum^{K}_{k=1} (m_k-1) \mathds{1}(t_{k}<t)$. Under this model, the density of observed statistics $\boldsymbol{m}$ and $\boldsymbol{t}$ is:
\begin{equation} \label{eq:lambda-coal-density}
    p(\boldsymbol{m,t}\mid \boldsymbol{n,s},N_{e}(t),\Lambda)=\exp\left(-\int^{\infty}_{0}\sum^{A(u)}_{j=2}\frac{\binom{A(u)}{j}\lambda_{A(u),j}}{N_{e}(u)}du \right)\prod^{K}_{k=1}\frac{\binom{A(t_{k})}{m_{k}}\lambda_{A(t_{k}),m_{k}}}{N_{e}(t_{k})}
\end{equation}
The density corresponds to the density of a point process of coalescent events with rate $\lambda_{A(t_k)} = \sum_{j=2}^{A(t_k)} \binom{A(t_{k})}{j} \lambda_{A(t_k), j}$  
when there are $A(t_k)$ lineages at time $t_{k}$, and the probability of $K$ coalescent events each being of size $m_k$: 
\begin{equation} \label{eq:density-block-size-only}
    \prod_{k=1}^K \frac{\binom{A(t_{k})}{m_{k}}\lambda_{A(t_{k}),m_{k}}}{\lambda_{A(t_k)}}.
\end{equation}
Under this MMC framework, we ignore the effects of population structure, recombination and selection. \\

In this manuscript, we are interested in estimating the effective population size $N_{e}(t)$ and the $\alpha$ parameter in the Beta$(2-\alpha, \alpha)-$coalescent with $\alpha\in (0,2]$ because it has desirable properties. As mentioned above, it is the limiting distribution of a specific Cannings population model. This model also converges to Kingman's coalescent model as $\alpha \to 2$ from below. \cite{hoscheit2019multifurcating} showed that the $\alpha$ parameter in the Beta-coalescent model is indicative of the level of superspreading in an infectious population. The mathematical basis of the Beta-coalescent for $\alpha \in (0,1]$ does not change, so for the sake of implementation, we allow $\alpha \in (0,2]$ in parameter estimation. However, most simulation results will focus on Beta-coalescent trees generated with $\alpha>1$. \\

Proposition~\ref{prop:upperbound} in Appendix~\ref{appendix:c} states that when there are $A(t)=b$ lineages, the total Lambda coalescent factor can be approximated as follows: 
\begin{equation} \label{eq:bound}
\lambda_{b}=\sum^{b}_{k=2}\binom{b}{k}\lambda_{b,k}\approx (b-1)\left(\frac{b}{2}\right)^{\alpha-1}. \end{equation}
That is, the logarithm of the total coalescent rate when there are $A(t)=b$ lineages is
\[ 
\log \left(\int^{t}_{0}\frac{\lambda_{b}du}{N_{e}(u)} \right) \approx (\alpha-1)\log(b/2)+ \log(b-1)+\log \left(\int^{t}_{0}\frac{du}{N_{e}(u)}\right). \]
Although we do not use this approximation for our inference methods, it provides intuition about the relationship between coalescent time distribution and its dependence on $\alpha$. \\

Interpretation of $N_{e}(t)$ requires knowledge of the population's variance in the number of offspring. In particular, under the Cannings population model, with variance of the number of offspring $c_{N}$, and assuming $c_N$ converges to 0 as $N\to\infty$, and the offspring distribution follows $\mathbb{P}(\nu_1 >x) \sim C x^{-\alpha}$ for $1\leq \alpha<2$ \citep{schweinsberg2003coalescent, freund2020cannings}, then the coalescent process converges to the Beta coalescent with \[N_{e}(t)=\lim_{N\to\infty} \frac{N(\lceil x/c_N \rceil)}{1/c_N}. \] In general, we assume populations whose effective population size is scaled in units of $1/c_{N}$ generations. We note that this factor depends on $\alpha$ in the case of the Cannings model, however, as long as $c_{N}$ remains constant over time, we can jointly estimate $N_{e}(t)$ and $\alpha$ for the Beta-coalescent and solve for $N(t)$ if $c_{N}$ is known. Alternatively, the exact parametric form of $c_{N}$ could be incorporated directly in Eq.~\ref{eq:lambda-coal-density}.

% \newpage 
\subsection{Estimation of $N_{e}(t)$ using a GMRF prior}
We follow the approach developed in \cite{pal12} and model  $N_e(t)= \exp[\gamma(t)]$, where $\gamma(t)$ is \textit{a priori} an intrinsic Gaussian Markov random field \citep{rue2005gaussian}, that is $\gamma(t)$ is a random piece-wise constant function with change points placed at a regular grid of $D$ points $\{x_1,..., x_D\}$. That is
\[N_{e}(t)=\sum_{d=1}^{D-1} \exp(\gamma_d)1_{(x_{d},x_{d+1}]}(t),\]
and $(\gamma_{1},\ldots,\gamma_{D-1})\sim MVN(0,(\tau Q)^{-1})$, $\tau$ is the precision parameter with Gamma prior $Gamma(0.001,0.001)$ and $Q$ is the corresponding inverse covariance kernel of a random walk with boundary correction. For more technical details, see \cite{pal12}. The integral in the exponent of Eq.~\ref{eq:lambda-coal-density} is then approximated by its Riemann sum. We adapted the implementation for binary trees in \texttt{phylodyn} \citep{karcher2017phylodyn} to the MMC with likelihood Eq.~\ref{eq:lambda-coal-density}. Given $\boldsymbol{m,s,t,n}$, and $\alpha$, our implemented \verb|R| function \verb|BNPR_Lambda()| estimates posterior mean and 95\% BCI of $N_{e}(t)$ by an Integrated Nested Laplace Approximation (INLA) \citep{rue2009approximate}.

\section{Methods} \label{sec:methods}

\textbf{Estimating $\alpha$ from tree topology only}: By the end of Section 2, we showed that $\alpha$ affects both the tree shape and the distribution of coalescent times. While coalescent times are the sufficient statistics for inferring $N_{e}(t)$ in the standard coalescent process, coalescent times and tree topology (i.e. merging block sizes) are the sufficient statistics for estimating $N_{e}(t)$ and $\alpha$ under the Beta-coalescent. However, it is possible that in some applications, there is no reliable estimation of coalescent times and it may be preferable to estimate $\alpha$ from the tree topology only. In this case, we propose to estimate $\alpha$ using just the block sizes $\boldsymbol{m}$ by maximizing the following pseudo-likelihood of the block sizes under the Beta$(2-\alpha, \alpha)$-coalescent (recall Equation~\ref{eq:density-block-size-only}):
\begin{equation}\label{eq:block_size_mle}
    \hat{\alpha}^{BS}=\argmax_{\alpha} \mathbb{P}(\boldsymbol{m} \mid \boldsymbol{n}, \alpha ) = \argmax_{\alpha} \prod^K_{k=1} \frac{\binom{A(t_k)}{m_k} \lambda_{A(t_k),m_k}(\alpha)}{ \sum_{i=2}^{A(t_k)} \binom{A(t_k)}{i} \lambda_{A(t_k),i }(\alpha)} 
\end{equation}
We know there is information about $\alpha$ in both the block sizes and the coalescent times. However, we show in Section~\ref{sec:results}, estimation of $\alpha$ based on tree topology only (block-sizes) is quite accurate. \\ 

\textbf{Hybrid estimation of $\alpha$ and $N_e(t)$}: Our first approach for estimating $\alpha$ and $N_{e}(t)$ from a multiple merger genealogy is a hybrid approach, where we iteratively update $\alpha$ and $N_e(t)$. We first initialize $\alpha^{(0)}= \hat{\alpha}^{BS}$, to be the block size only MLE. Given $\alpha$, we then use INLA, as described in Section 2.2, and update $N_e(t)$ by the posterior median. Given the current value of $N_{e}(t)$, we then find $\alpha$ by maximum likelihood estimation, and iterate these last two steps until convergence. \\ 

\textbf{Joint posterior inference of $N_{e}(t)$ and $\alpha$}: We place a GMRF prior on $\log N_{e}(t)$ as described in Section 2.2 with a Gamma prior on precision parameter $\tau$, and a uniform $U(0,2)$ prior on $\alpha$. We approximate $\mathbb{P}(\alpha, N_e(t), \tau \mid \boldsymbol{g})$ via a Metropolis-within-Gibbs algorithm, where $\boldsymbol{g}$ denotes the multifurcating genealogy. The precision parameter and the GMRF are usually highly correlated leading Markov chain Monte Carlo methods to have poor convergence and slow mixing. To solve this problem, we sample from the conditional distribution of $\log N_{e}(t)$ and $\tau$ given all other model parameters and genealogy, with the split Hamiltonian Monte-Carlo (sHMC) method adapted from \cite{lan2015efficient}. To sample $\alpha$ from the full marginal 
\begin{equation} \label{eq:alpha_cond}
        \mathbb{P}(\alpha \mid \boldsymbol{g}, N_e^{(k)}(t),\tau) =  \frac{\mathbb{P} (\boldsymbol{g} \mid  N_e^{(k)}(t), \alpha)}{\int_0^2 \mathbb{P} (\boldsymbol{g} \mid  N_e^{(k)}(t), \alpha) d\alpha},
\end{equation}
we discretize the parameter space of $\alpha \in [0,2]$ into intervals of length $0.005$ and approximate the denominator of Eq.~\ref{eq:alpha_cond} by the Riemann sum. That is, we set $\alpha_m$ to be the midpoint of interval $I_m=[0.005m, 0.005(m+1)]$ for $m=1,...,400$. The numerator of Equation~\ref{eq:alpha_cond} is proportional to the likelihood of Equation~\ref{eq:lambda-coal-density}, and the proportionality constant cancels in the fraction. Then 
\begin{equation*}
    \mathbb{P}(\alpha \in I_m \mid  \boldsymbol{g}, N_e^{(k)}(t)) \approx \frac{\mathbb{P} (\boldsymbol{g} \mid  N_e^{(k)}(t), \alpha_m)}{\sum_{i=1}^{400} \mathbb{P} (\boldsymbol{g} \mid  N_e^{(k)}(t), \alpha_i)}.
\end{equation*}
From this discrete distribution, we sample $I_m \sim \mathbb{P}(\alpha \in I_m \mid  \boldsymbol{g}, N_e^{(k)}(t)), \alpha^{(k+1)} \sim Unif[I_m]$ to get our final update. 

\section{Results} \label{sec:results}
We first test our methods on simulated data. We implement a function that can simulate a Beta-coalescent genealogy under variable effective population size and heterochronous sampling. We use the true genealogy as input in simulation studies. To evaluate their performance in estimating $\alpha$, we use the mean squared error (MSE). For evaluating the performance in estimating the effective population size, we compute four metrics: coverage, bias, deviance, and MSE at the grid points, with precise definitions in Appendix~\ref{appendix:a}. \\

To see how well we can estimate $\alpha$ from tree topology alone (block-size data), we generated 1000 realizations from the Beta-coalescent under different values of $\alpha, n$. Table~\ref{tab:topology_only_mle} in Appendix~\ref{appendix:b} shows the mean, bias, and MSE of the estimated $\alpha$ using only the block-size data. Overall, we see $\hat{\alpha}^{BS}$ underestimates $\alpha$ and the accuracy of the estimator increases with $n$ and $\alpha$. When $\alpha=1.5$, the accuracy of this estimator is quite high for trees with 100 tips. When $\alpha=1.8$, the results are good for trees with 50 tips. \\ 

In the left plot in Figure~\ref{fig:sim_tree_ex}, we show one isochronous tree with 50 tips, generated with $\alpha=1.5$ and exponential growth $N_{e}(t)=1000e^{-t}$. The block-size MLE is $\hat{\alpha}^{BS}\approx 1.615$ and the hybrid estimate is $\hat{\alpha}^{H} \approx 1.759$. The posterior median and mean of $\alpha$ inferred from MCMC are $1.594$ and $1.585$ respectively, which are the closest to the true $\alpha$ value. Figure~\ref{fig:one_tree_posterior_alpha_dist} in Appendix~\ref{appendix:b} shows the posterior distribution of $\alpha$ from MCMC. The right plot in Figure~\ref{fig:sim_tree_ex} depicts the reconstructed effective population size trajectories and their 95\% credible regions estimated with each method, as well as the true trajectory. Table~\ref{tab:sim_tree_ex_perf_meas} in Appendix~\ref{appendix:b} details the performance measures calculated for inferred $N_e(t)$ with each method. We see that the inferring $N_e(t)$ using the true $\alpha$ has the lowest bias, deviance, and MSE, while out of the three proposed methods, MCMC performs the best and hybrid estimation the worst. This is expected since the $\hat{\alpha}^H$ is the worst estimate. In Appendix~\ref{appendix:b} Figure~\ref{fig:one_tree_trace_plot}, we show trace plots of $\alpha$ and $\log \; N_e(t)$ at a particular value of $t$ to demonstrate well-mixing. \\

\begin{figure}[h]
\setcounter{subfigure}{0}
    \centering
    \sidesubfloat[]{\includegraphics[width=0.36\linewidth]{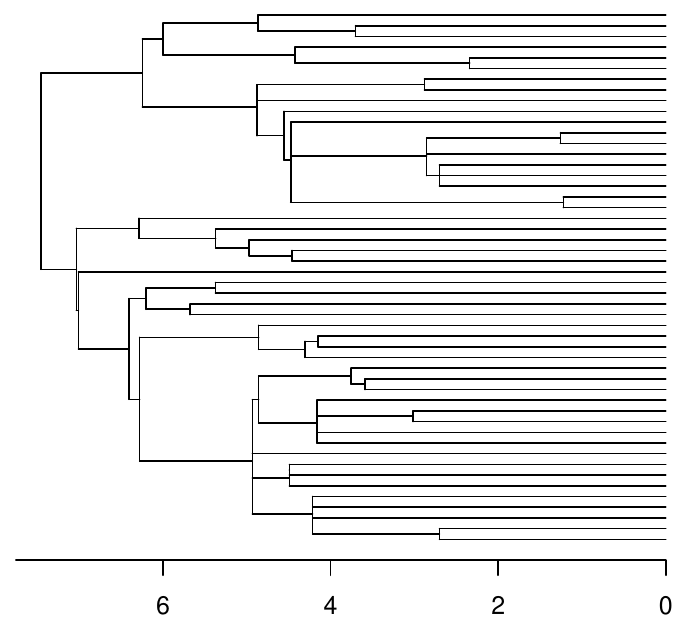}}
    \sidesubfloat[]{\includegraphics[width=0.54\linewidth]{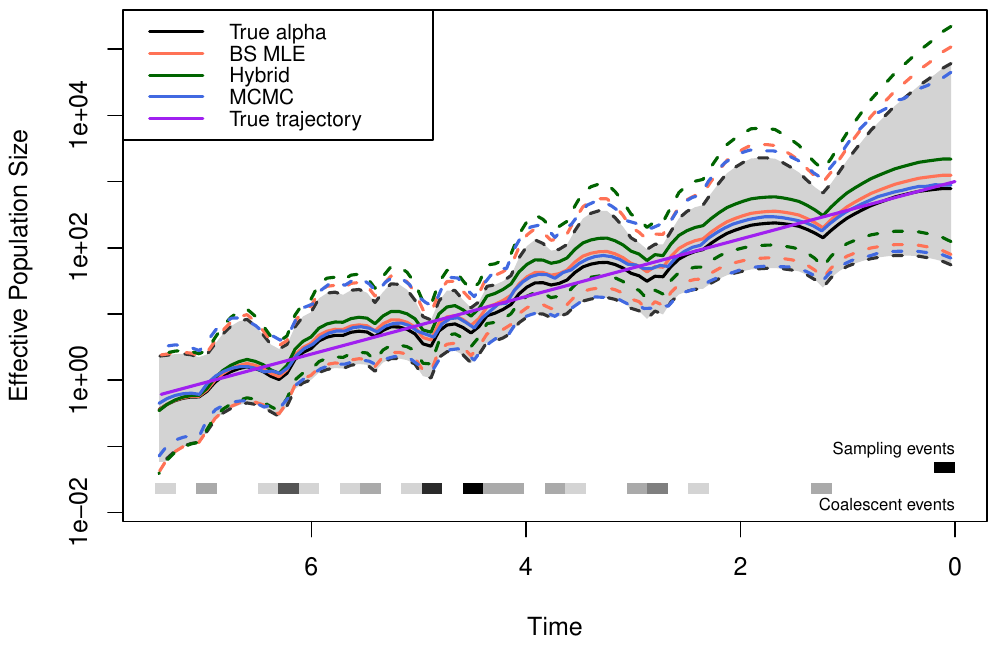}}
    \caption{\textbf{Comparison of different methods estimating  exponential $N_{e}(t)$.} \textbf{(a):} A simulated genealogy under the Beta-coalescent with $\alpha=1.5$ and the exponential growth $N_{e}(t)$ trajectory. The genealogy has 39 internal nodes and 50 tips. \textbf{(b):} The effective population size trajectories reconstructed by different methods. Note the y-axis is plotted in log-scaled. Solid lines are the median trajectories, and dotted lines are the 95\% credible bands. The true trajectory is plotted in purple, and the solid black line and shaded band are the median and 95\% credible band of the inferred $N_e(t)$ using the true $\alpha=1.5$.}
    \label{fig:sim_tree_ex}
\end{figure}

To evaluate the overall performance of our methods for the joint inference of $\alpha$ and $N_{e}(t)$, we simulated 50 trees for each combination of $\alpha \in \{1,1.5,1.8\}$, $n \in \{20, 50, 100\}$ with three sampling schedules: (1) all sampled at time 0 (isochronous), (2) heterochronous with 2 sampling times, 50\% split, and (3) heterochronous with 4 sampling times, split 50\%, 30\%, 10\% and 10\% across the 4 times. We also considered three effective population size trajectories: (1) Uniform (constant): $N_e(t) =100$, (2) Exponential growth: $N_e(t)=1000e^{-t}$, and (3) Boom-bust: $N_e(t)= 1000 e^{-|t-1|}$. \\

For the MCMC method, we run 20,000 iterations discarding the first 10\% as burn-in. For each tree, we obtain the following set of results for $\alpha$: the block-size $\hat{\alpha}^{BS}$, the estimate from the hybrid method $\hat{\alpha}^H$, and a posterior distribution of $\alpha$ estimated via MCMC. We also infer the effective population size at a regular grid of 100 points using INLA and the true $\alpha$ value, the two estimates $\hat{\alpha}^{BS}$ and $\hat{\alpha}^H$, as well as the posterior distribution of $N_e(t)$ obtained by MCMC. \\ 

Figure~\ref{fig:all_traj_alpha_boxplots} in Appendix~\ref{appendix:b} shows the boxplots of estimated $\alpha$ values for each simulated genealogy with three big blocks corresponding to the three $N_{e}(t)$ scenarios considered. The posterior median is used as the final $\alpha$ estimate in the MCMC case. Each plot corresponds to different $\alpha,n$ combinations, with the number of tips $n$ varying across rows and $\alpha$ values varying across columns. The three different sampling schedules are indicated by the color of the boxplots. The black line shows the true $\alpha$ value. First, we see the results do not vary much across the different sampling times, which is a good sign of robustness, and so going forward will pool results from the three different sampling times together. We also see that accuracy increases with $\alpha$ and the number of tips $n$. \\

\begin{table}[!h]
\centering
\begin{tabular}{cclccccccccccc}
\hline

&& & \multicolumn{3}{c}{\multirow{2}{*}{\textbf{Coverage}}} && \multicolumn{3}{c}{\multirow{2}{*}{\textbf{Deviance}}} & & \multicolumn{3}{c}{\multirow{2}{*}{\textbf{MSE}}} \\
&&& \multicolumn{2}{c}{} & \multicolumn{2}{c}{} & \multicolumn{2}{c}{} & \multicolumn{2}{c}{}  \\ \hline
& $N$ & Method & Unif. & Exp. & BB  && Unif. & Exp. & BB  & & Unif. & Exp. & BB \\ 
  \hline
\multirow{9}{*}{$\alpha=1$} & \multirow{3}{*}{20} & BS MLE & 0.92& 0.96 & 0.95  &  & \textbf{0.4} & 1.23 & 1.57 && \textbf{19.34 }& 117.14 & 287.95\\ 
  & & Hybrid &0.91& 0.84 & 0.84 &&  0.55 & 3.76 & 5.03 & & 31.7 & 1796.08 & 8197.65 \\ 
  & & MCMC & \textbf{0.99} & \textbf{0.98} & \textbf{0.98} &&0.54 & \textbf{0.8} & \textbf{0.81} && 34.59 & \textbf{94.34} & \textbf{118.74}\\ \cline{2-14}
  & \multirow{3}{*}{50} & BS MLE & 0.9 & 0.94 & 0.93 && 0.37 & 0.74 & 0.82 & & 15.14 & 78.4 & 156.55  \\ 
  & & Hybrid & 0.9 & 0.76 & 0.79 && \textbf{0.356} & 2.463 & 2.937 & & \textbf{15.101} & 627.103 & 2389.601 \\ 
  & & MCMC & \textbf{0.97} & \textbf{0.97} &\textbf{0.97} && 0.42 & \textbf{0.61} & \textbf{0.71} & & 25.6 & \textbf{74} & \textbf{132.66}  \\ \cline{2-14}
  & \multirow{3}{*}{100} & BS MLE & 0.89 & 0.95 & 0.97 & &\textbf{0.32} & 0.56 & 0.63 &  & \textbf{12.49} & \textbf{42.99} & 93.95 \\ 
  & & Hybrid & 0.88 & 0.72 & 0.76 &&0.35 & 2.21 & 2.66 & & 15.02 & 671.61 & 2454.44 \\ 
  & & MCMC & \textbf{0.98} & \textbf{0.99} & \textbf{0.98} && 0.393 &\textbf{ 0.5} & \textbf{0.52 } &  & 20.06 & 49.37 & \textbf{88.1}\\ \cline{2-14}
  \multirow{9}{*}{$\alpha=1.5$} & \multirow{3}{*}{20} & BS MLE & 0.95 & 0.97 & 0.96 && \textbf{0.32} & 0.79 & 0.88 &  & \textbf{11.36} & 63.39 & 130.03 \\ 
  & & Hybrid & 0.95 & 0.94 & 0.95 && 0.33 & 1.28 & 1.67 & & 11.6 & 121.85 & 677.61 \\ 
  & & MCMC & \textbf{0.99 }&\textbf{ 0.97} & \textbf{0.98 }&& 0.33 & \textbf{0.54} & \textbf{0.57 } & & 14.22 & \textbf{63.44} & \textbf{85.45} \\ \cline{2-14}
  & \multirow{3}{*}{50} & BS MLE & 0.95 & 0.97 & 0.96 && \textbf{0.26} & 0.66 & 0.68 & & \textbf{7.47} & 60.97 & 125.58 \\ 
  & & Hybrid & 0.95 & 0.94 & 0.95 && 0.27 & 1.05 & 1.12 & & 8.13 & 115.94 & 497 \\ 
  & & MCMC & \textbf{0.98} &\textbf{ 0.99} & \textbf{0.98} && 0.34 & \textbf{0.58} & \textbf{0.57} & & 14.71 & \textbf{52.19} & \textbf{85.16}\\ \cline{2-14}
  & \multirow{3}{*}{100} & BS MLE & 0.95 & 0.94 & 0.95 && \textbf{0.21} & 0.43 & 0.46 & & \textbf{5.24} & 36.06 & 50.33 \\ 
  & & Hybrid & 0.95 & 0.85 & 0.87 && 0.21 & 0.74 & 0.81 & & 5.33 & 71.44 & 193.77 \\ 
  & & MCMC &\textbf{ 0.99 }&\textbf{ 0.98} & \textbf{0.98} && 0.28 & \textbf{0.41} & \textbf{0.41} & & 9.84 & \textbf{34.78} & \textbf{49.26} \\ \cline{2-14}
  \multirow{9}{*}{$\alpha=1.8$} & \multirow{3}{*}{20} & BS MLE & 0.98 & 0.98 & 0.98 &&\textbf{0.23} & 0.66 & 0.74 & & 6.44 & 60.77 & 153.67\\ 
  & & Hybrid & 0.99 & 0.97 & 0.98 && 0.24 & 0.74 & 0.86 & & \textbf{6.3} &\textbf{ 56.85 }& 207.42 \\ 
  & & MCMC & \textbf{1} & \textbf{0.98} &\textbf{ 0.99} && 0.34 & \textbf{0.5} & \textbf{0.57} & & 14.7 & 59.86 & \textbf{86.87} \\ \cline{2-14}
  & \multirow{3}{*}{50} & BS MLE & 0.98 & 0.98 & 0.98 &&\textbf{ 0.2 }& 0.44 & 0.49 & & \textbf{4.25} & 36.24 & 81.19 \\ 
  & & Hybrid & 0.98 & 0.98 & 0.98 && 0.22 & 0.56 & 0.64 & & 5 & 42.94 & 129.35 \\ 
  & & MCMC & \textbf{0.99} & \textbf{0.99} & \textbf{0.99 }&& 0.29 & \textbf{0.41 }& \textbf{0.41 } & & 10.25 & \textbf{27.56} & \textbf{56.11} \\ \cline{2-14}
  & \multirow{3}{*}{100} & BS MLE & 0.97 & 0.96 & 0.96 && 0.17 & 0.43 & 0.43 & & 3.6 & 23.18 & 37.24 \\ 
  & & Hybrid & 0.97 & 0.96 & 0.97 &&\textbf{0.16} & 0.48 & 0.47 & & \textbf{3.27} & 25.03 & 52.15 \\ 
  & & MCMC & \textbf{0.99} & \textbf{0.97} & \textbf{0.98} && 0.24 & \textbf{0.4} & \textbf{0.39 } &  & 7.7 & \textbf{22.89} & \textbf{31.76}\\ 
   \hline
\end{tabular}
\caption{Mean coverage, median deviance, and median MSE of estimated $N_e(t)$ for each trajectory and each method. The method with the best performance per $N,\alpha$ and trajectory combination is shown in boldface.}
\label{tab:median_deviance_mse_Ne}
\end{table} 

MCMC is the top performing method according to deviance for estimating $\alpha$ for over 40\% of the simulations, with block size MLE and hybrid estimation have similar performances (see Table~\ref{tab:best_method_alpha} in Appendix~\ref{appendix:b}). Table~\ref{tab:median_deviance_mse_Ne} shows the summary of performance statistics for estimating $N_{e}(t)$ (with the different sampling times combined). The value in boldface is the best performing (largest coverage, smallest deviance, and smallest MSE) out of the three methods. In terms of coverage, the MCMC method has the best performance, however all methods in general have good coverage. In terms of median deviance and mean MSE, the MCMC method is highly superior for the exponential and boom-and-boost trajectories, however we see the block-size MLE performs the best for the uniform trajectory, with the hybrid method not far behind. One possible explanation is a lack of identifiability in the case of constant $N_e$: scaling $\alpha$ by $c$ and $N_e$ be $1/c$ would result in the same likelihood. \\

With regards to the computational cost of the methods, one run of the MCMC method (for 20,000 iterations) on a multifurcating tree with 100 tips takes approximately 15 minutes on one CPU node, while the BS MLE method takes less than 10 seconds, and the hybrid method takes 3 minutes. Computational complexity of all methods is roughly linear in $n$, the number of tips.

\newpage 
\section{Applications} \label{sec:applications}
We apply our methods to two infectious disease examples and one animal population, where multifurcating viral phylogenies may occur due to superspreading events or skewed reproduction. All data used here can be downloaded directly from the open source dashboard Nextstrain \citep{hadfield2018nextstrain}, or from the GenBank database \citep{benson2012genbank}. We generated point estimates of the genealogies ignoring uncertainty and the quality of the estimation procedure. Our methods are not robust against misestimations of the genealogy (see Appendix~\ref{appendix:reconstruction}), and incorporation of genealogical uncertainty is still an open problem.

\subsection{Respiratory Syncytial Virus (RSV)} 
RSV is a respiratory illness transmitted via droplets that mostly affects children and immunocompromised individuals. Globally in 2019, there is an estimated 33 million RSV-associated acute lower respiratory infection episodes (ARTI) in children under the age of five \citep{li2022global}. These ARTI episodes mostly result in hospitalizations and health-care burdens, and current surveillance efforts are not well-developed, especially in lower income or lower-middle income countries such as those in Southeast Asia \citep{divarathne2019impact, pebody2020approaches}. Research also suggests that RSV may be seasonal, with more than 75\% of cases within five months, mainly occurring during the rainy season in tropical countries \citep{weber1998respiratory, thongpan2020respiratory}. In terms of the virology, RSV is a negative-sense, single-stranded RNA virus with 10 genes that encode 11 proteins \citep{griffiths2017respiratory}. In particular, two major surface proteins (F and G glycoproteins) control viral attachment and are the primary antibody targets, with the G protein being highly variable. There are also two antigenic subtypes based on the reaction of F and G proteins to monoclonal antibodies: RSV-A and RSV-B. Research shows that RSV-A is more prevalent and virulent than RSV-B, making it more concerning in terms of public health \citep{falsey2000respiratory, jha2016respiratory}. \\

Here, we consider the genealogy shown in Figure~\ref{fig:rsv_asia_ex}(a) estimated from 266 RSV-A molecular sequences\footnote{The data can be downloaded at \url{https://nextstrain.org/rsv/a/G?f_country=Bangladesh,Cambodia,India,Laos,Myanmar,Nepal,Thailand,Vietnam} sampled in South and Southeast Asia (mainly in India and Thailand) from 1991 to 2023. Sequences were aligned with MAFFT \citep{katoh2002mafft} and the genealogy was reconstructed with IQ-Tree \citep{nguyen2015iq}. Finally, we used TreeTime \citep{sagulenko2018treetime} for temporal inference and branches of length zero were collapsed.} We accessed the webpage on April 9, 2024. The block-size MLE is $\hat{\alpha}^{BS} = 1.718$ and the hybrid estimate is $\hat{\alpha}^H = 1.734$, and the MCMC posterior median and mean are 1.698 and 1.696. The inferred effective population size trajectories are shown in Figure~\ref{fig:rsv_asia_ex}(b). We see a rather steady increase until around 2010, with some oscillations indicative of the seasonality of RSV. The credible intervals obtained with MCMC are overall narrower than those obtained with the other two methods. \\ 
 
\begin{figure}[h]
\setcounter{subfigure}{0}
    \centering
    \sidesubfloat[]{\includegraphics[width=0.46\linewidth]{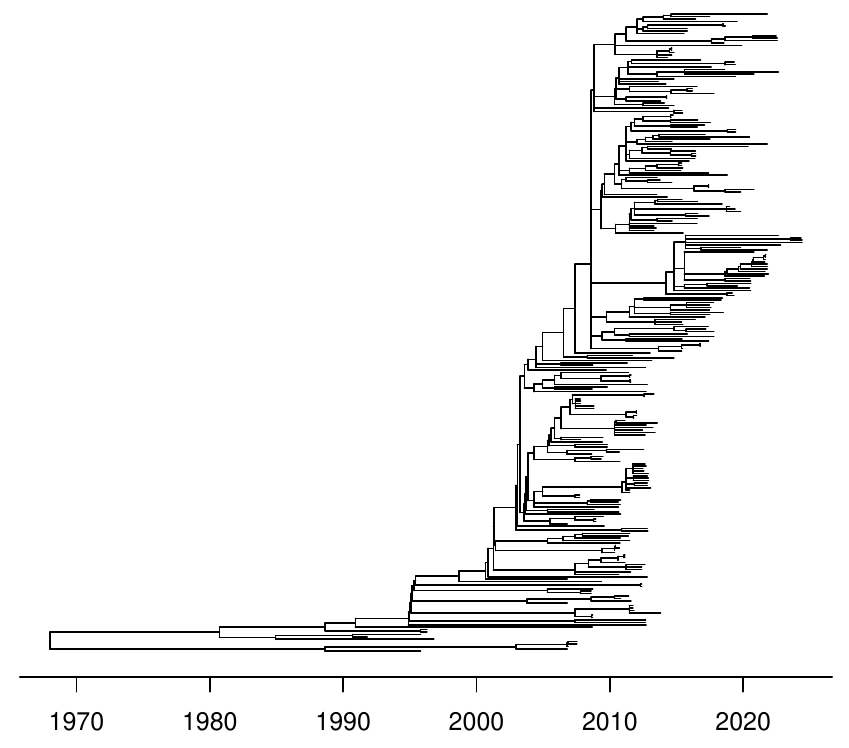}}
    \sidesubfloat[]{\includegraphics[width=0.48\linewidth]{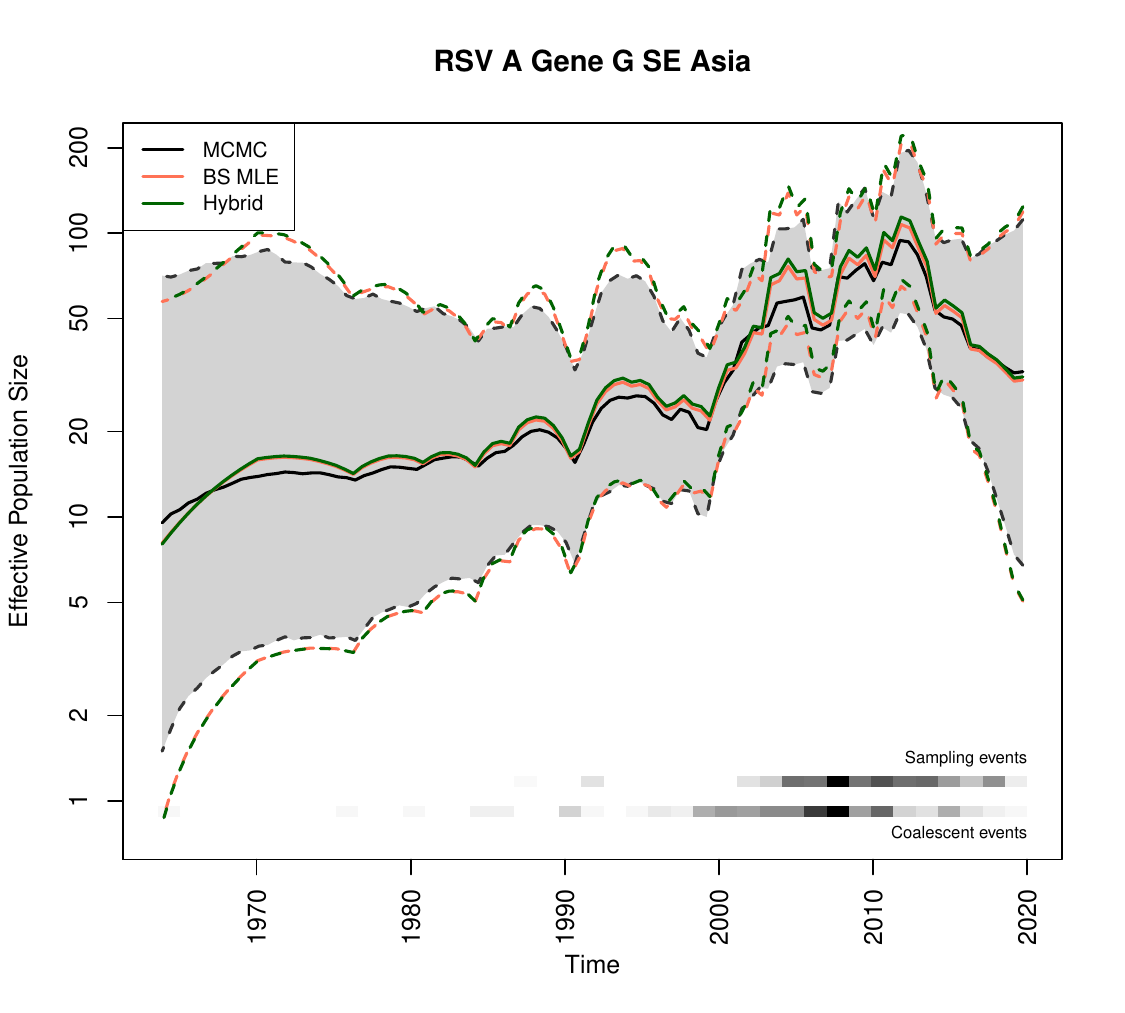}}
    \caption{\textbf{(a):} A dated genealogy generated using 266 RSV-A sequences from Nextstrain. \textbf{(b):} The reconstructed effective population size trajectories for RSV-A sequences in South and Southeast Asia. Note the $y$-axis is plotted in log-scale. Solid lines are the median trajectory, and dotted lines are the 95\% credible band.}
    \label{fig:rsv_asia_ex}
\end{figure}

\subsection{Enterovirus D68}
First discovered in 1962, Enterovirus D68 is a non-polio enterovirus that causes respiratory illnesses. It is spread person-to-person via respiratory secretions and mainly affects children. Enterovirus D68 has been found to be associated with acute flaccid myelitis (AFM), a muscle weakness condition which can lead to paralysis, but not caused by polio \citep{dyda2018association, moline2019notes}. There is no treatment or vaccine and so large outbreaks are important public health concerns. Prior to 2014, there were only 699 confirmed cases, but since then, outbreaks in the US, Canada, Europe, and Asia have regularly occurred. Over 2000 cases were confirmed worldwide in 2014 alone \citep{holm2016global}. Since diagnoses can only be confirmed via laboratory testing, these case counts are likely to be underestimated. Cases are also seasonal, occurring mainly from August to October, with the US having had outbreaks on a two-year cycle in 2014, 2016, and 2018 \citep{centers2011clusters, mckay2018increase}. \\

Post-COVID re-opening saw an outbreak of 139 cases in eight European countries between July and October 2021 \citep{benschop2021re}. To investigate the spread of Enterovirus D68 in Europe prior to COVID, we analyze a genealogy (Figure~\ref{fig:ent_europe_ex}(a)) inferred from 171 European sequences\footnote{The data used can be downloaded at \url{https://nextstrain.org/enterovirus/d68/genome?f_region=Europe}. The analysis was performed in the same manner as the RSV example and we accessed the webpage on April 9, 2024.} sampled between 1998 and 2019. Most samples are of children under the age of 5 and belong to clade B3, which was the dominant strain of the 2016 enterovirus D68 outbreak in the US \citep{wang2017enterovirus}. The block-size MLE is $\hat{\alpha}^{BS} = 1.741$, the hybrid estimate is $\hat{\alpha}^H = 1.764$. The MCMC posterior median and mean from MCMC are 1.732 and 1.728. The inferred effective population size trajectories are depicted in Figure~\ref{fig:ent_europe_ex}(b).  We see an oscillatory pattern, with a peak prior to 1998, then rather constant until a decrease in 2008 and then peak around 2012. Then we see peaks around 2015 and 2019 which is consistent with enterovirus outbreaks observed clinically, if we account for delays in sequencing. Again, the block-size MLE method and hybrid method give slightly higher estimates than MCMC. The credible intervals for MCMC are slightly narrower than the other two methods, but all three give similar results. \\

\begin{figure}[h]
\setcounter{subfigure}{0}
    \centering
    \sidesubfloat[]{\includegraphics[width=0.46\linewidth]{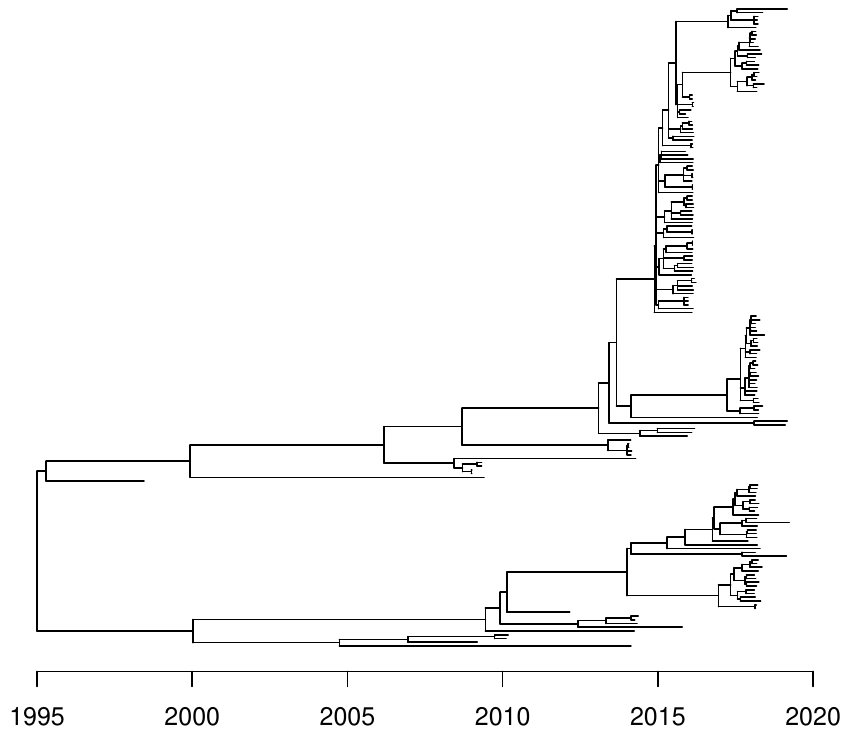}}
    \sidesubfloat[]{\includegraphics[width=0.48\linewidth]{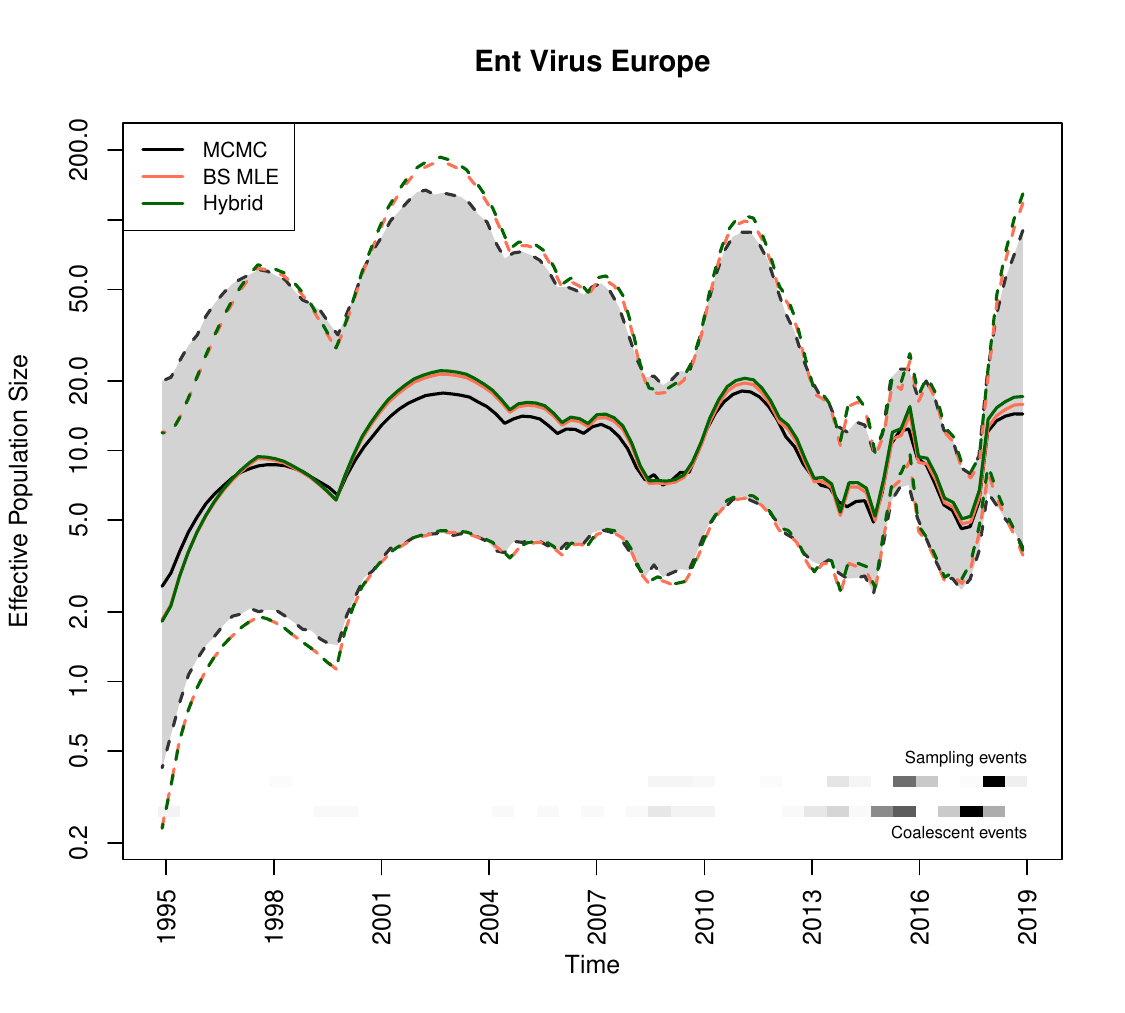}}
    \caption{\textbf{(a):} A dated phylogeny generated using 171 enterovirus D68 sequences from Nextstrain. \textbf{(b):} The reconstructed effective population size trajectories for Enterovirus D68 sequences in Europe. Note the $y$-axis is plotted in log-scale. Solid lines are the median trajectory, and dotted lines are the 95\% credible band.}
    \label{fig:ent_europe_ex}
\end{figure}

\subsection{Japanese Sardine populations}
Three hypotheses have been postulated as possible explanations for an excess of low-frequency mutations in many marine species, such as sardines. The first is that populations experienced a population bottleneck followed by a period of expansion during and after the Last Glacial Maximum. The second hypothesis states that species experienced natural selection and selective sweeps. The third is the reproductive skew hypothesis: there exists high variation in the number of offspring produced by many individuals which is then reflected in the genetic diversity. \cite{niwa2016reproductive} analyzed 106 mitochondrial DNA sequences of Japanese sardine (\textit{Sardinops melanostictus}) collected in 1990, where they conclude the data strongly supports the presence of multiple mergers. They assumed the Beta$(2-\alpha, \alpha)-$coalescent model as the underlying MMC model and used an importance sampling scheme with the infinite-sites mutation model to reject the hypothesis of $\alpha=2$ \citep{birkner2011importance}, with $\alpha$ estimated to be approximately 1.3 by maximum likelihood. Further results in their analysis reject the demographic expansion hypothesis using site-frequency spectrum and pairwise nucleotide differences. Understanding the presence of reproductive skew is important in order to better estimate the current and historical effective population size of marine species. \\

We accessed the same mtDNA sequences from GenBank (accession nos. LC031518–LC031673), and reanalyzed the reproductive skew hypothesis with our methods. Out of the 156 total sequences, 106 were collected in 1990 and used in both of the previous papers for analysis, while the other 50 sequences, collected between 2010-2012 by the same authors were not included in the original study. We first used MAFFT to align the sequences of about 1200bp of length, then used IQ-Tree Version 1.6.12 and TreeTime to infer dated, multifurcating trees via maximum likelihood \citep{katoh2002mafft, nguyen2015iq, sagulenko2018treetime}. Again, branches of length zero were collapsed. We fix the mutation rate in TreeTime to be $2\times 10^{-7}$ per site per year. The related Indian oil sardine, \textit{Sardinella longiceps}, estimated a mutation rate of $1\times 10^{-7}$ per site per year with a strict molecular clock \citep{sukumaran2016population}. Figure~\ref{fig:sardines_tree}(a) and (b) show the inferred trees using 156 sequences and the sub-tree containing only sequences from 1990 respectively. The dated root is inferred to be around 29000 years ago, which roughly coincides around the time of the last glacial maximum. \\

\begin{figure}[h]
    \centering
    \includegraphics[width=0.95\linewidth]{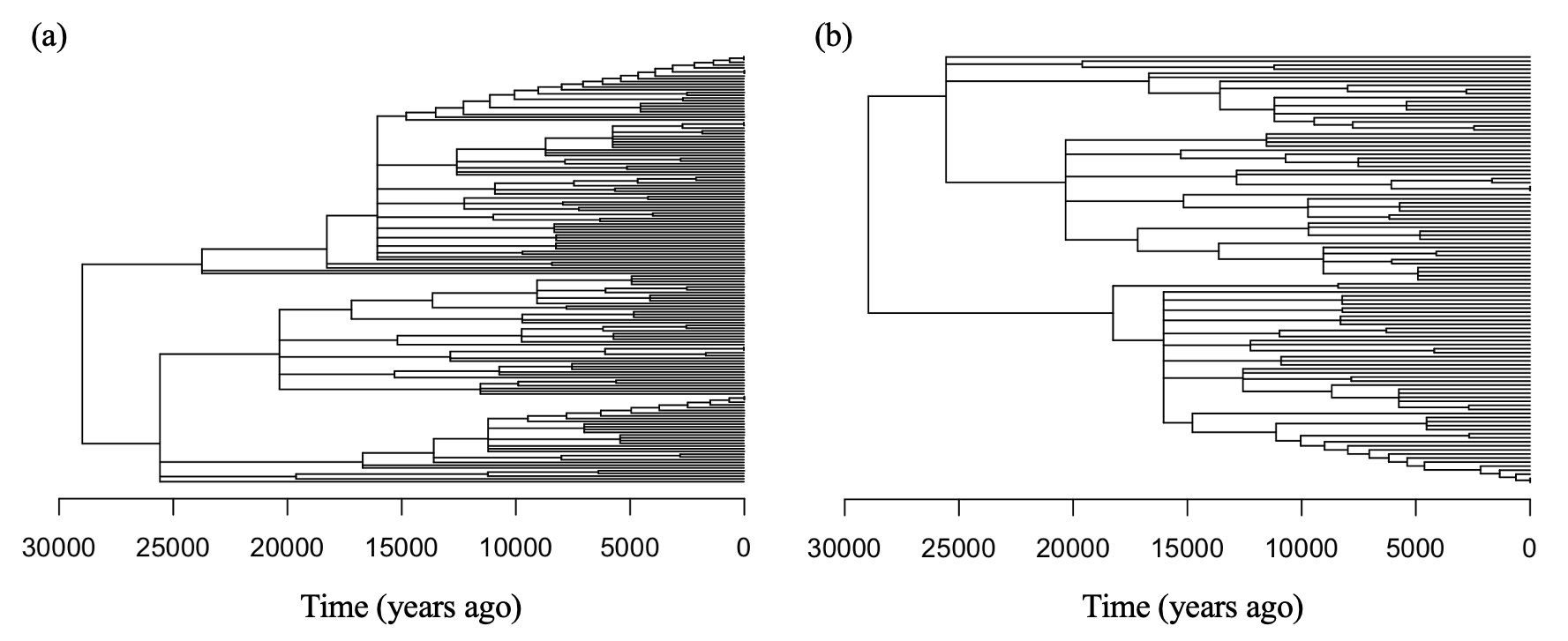}
    \caption{\textbf{(a):} A dated phylogeny generated using all \textit{Sardinops melanostictus} sequences. There are 156 tips and 95 internal nodes. \textbf{(b):} The subtree constructed of only sequences sampled in 1990: 106 tips and 61 internal nodes.}
    \label{fig:sardines_tree}
\end{figure}

We first estimate model parameters from the subtree with only sequences from 1990. The block-size MLE is $\hat{\alpha}^{BS} = 1.3498$, the hybrid estimate is $\hat{\alpha}^H = 1.3855$, and the MCMC posterior median and mean are 1.319 and 1.317 respectively. These inferred $\alpha$ values are largely consistent with those reported in \cite{niwa2016reproductive} that used an alternative method. The inferred effective population size trajectories depicted in Figure~\ref{fig:sardines_both_bnpr}(a) show a slow increase in $N_{e}(t)$ over time, however a constant trajectory falls within the 95\% credible intervals of the MCMC posterior. The MCMC estimates are less smooth than the block-size MLE and the hybrid estimates. Analyzing all sequences, we obtain the block-size MLE $\hat{\alpha}^{BS} = 1.4119$, the hybrid estimate is $\hat{\alpha}^H = 1.4124$, and the MCMC posterior median and mean of 1.3686 and 1.3659 respectively. Adding in extra data has reduced the variability in the MCMC method, and also increased the inferred $\alpha$ values. The inferred effective population size trajectories are shown in Figure~\ref{fig:sardines_both_bnpr}(b). Note the block-size MLE trajectory is covered by the hybrid estimate one because their $\alpha$ values are very close, but they are not identical. The inferred trajectories with all sequences are very similar to those inferred from only the sequences collected in 1990. \\

\begin{figure}[h]
\setcounter{subfigure}{0}
    \centering
    \sidesubfloat[]{\includegraphics[width=0.47\linewidth]{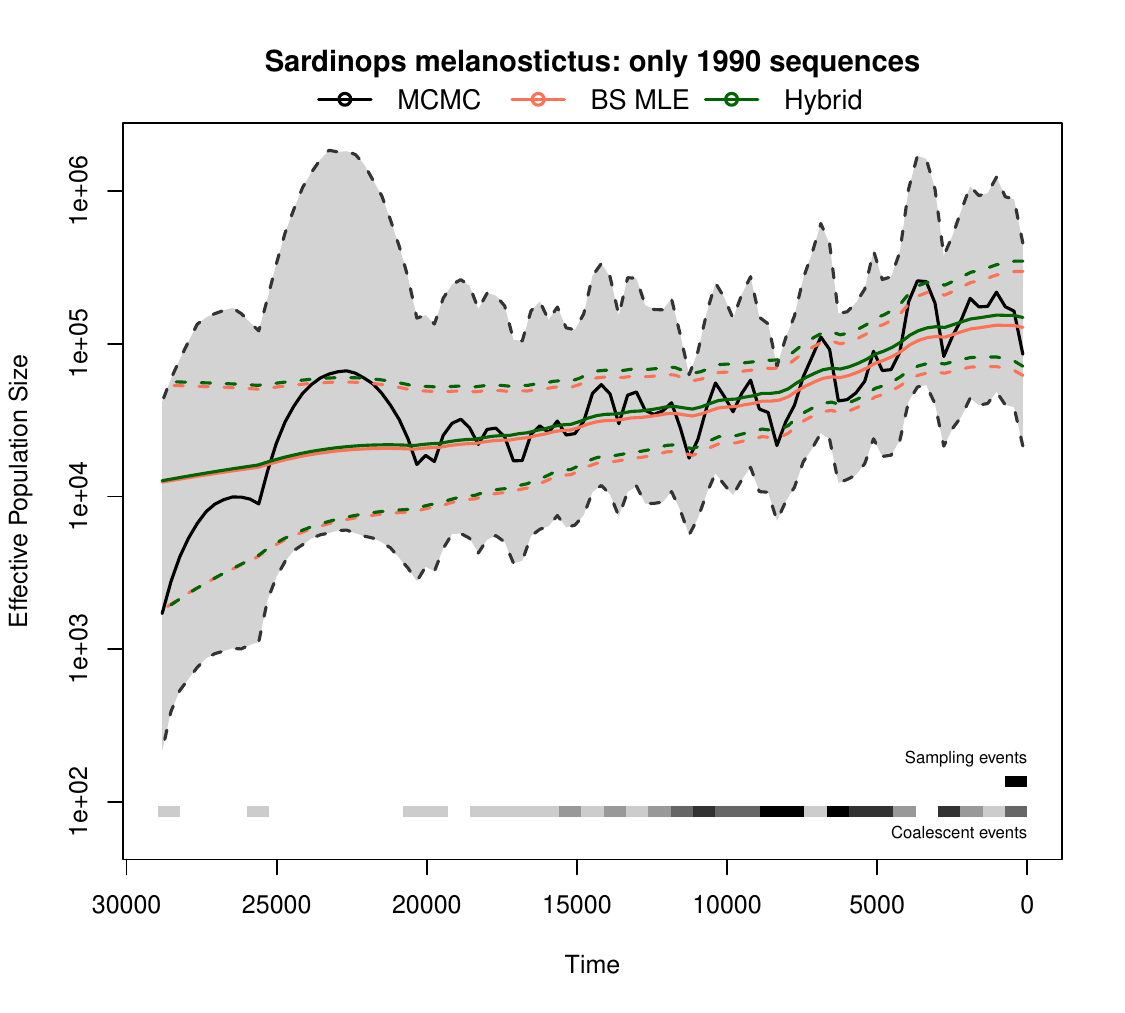}}
    \sidesubfloat[]{\includegraphics[width=0.47\linewidth]{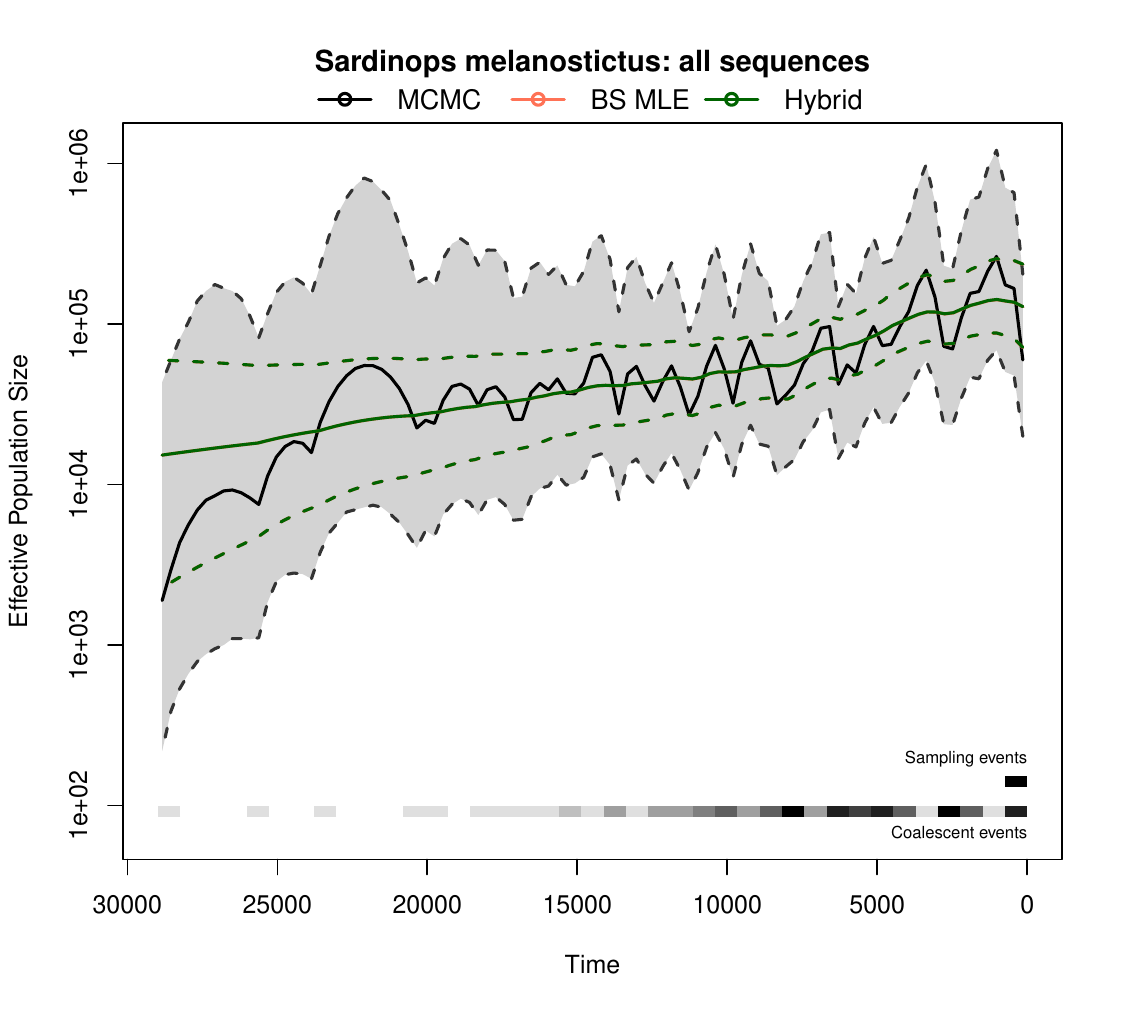}}
    \caption{The reconstructed effective population size trajectories for (a): the 106 sequences sampled in 1990 and (b): all \textit{Sardinops melanostictus} sequences. Note the $y$-axis is plotted in log-scale. Solid lines are the median trajectory, and dotted lines are the 95\% credible band.}
    \label{fig:sardines_both_bnpr}
\end{figure}

\cite{matuszewski2018coalescent} used site frequency spectrum (SFS) based maximum likelihood methods under the $\psi-$coalescent and exponential population growth to jointly investigate reproductive skew and population expansion hypotheses.  The $\psi-$coalescent is a $\Lambda$-coalescent with measure $\Lambda= \delta_\psi, \psi \in [0,1]$, and is the limiting distribution of an extended, discrete-time Moran model \citep{eldon2006coalescent}. The authors estimate $\hat{\psi}=0.46$, indicating reproductive skew, but no exponential population growth. Overall, our results agree with this study, except that out detailed inference of $N_{e}(t)$ shows a slow increase over time. \\ 

\begin{figure}[h]
\setcounter{subfigure}{0}   
    \centering
    \sidesubfloat[]{\includegraphics[width=0.4\linewidth]{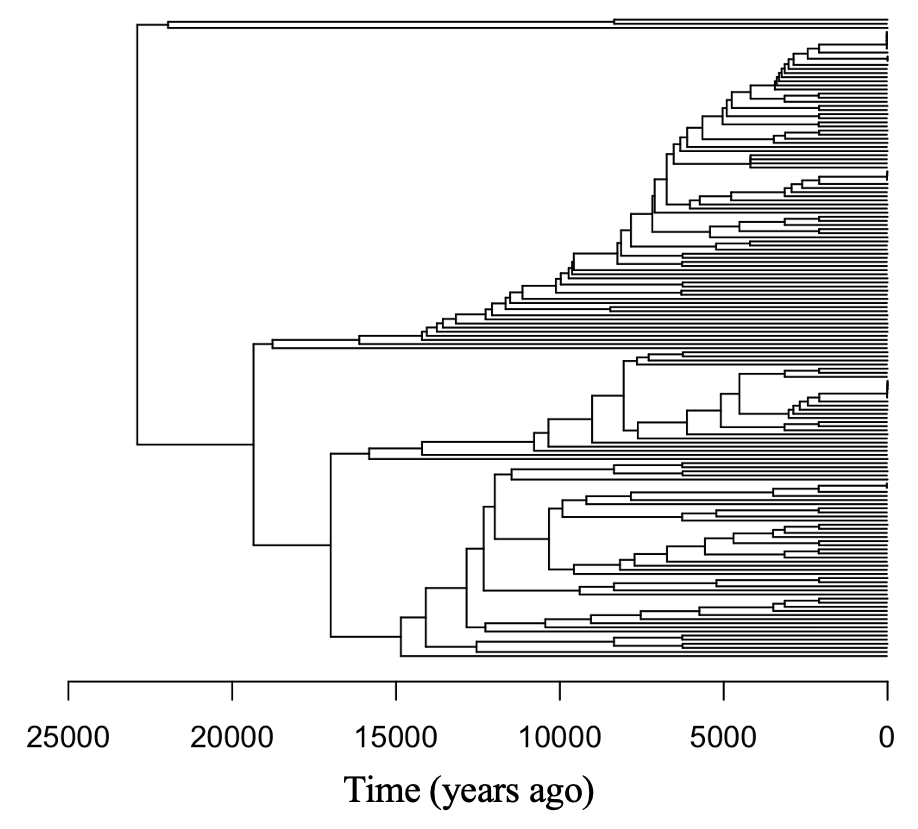}}
    \sidesubfloat[]{\includegraphics[width=0.52\linewidth]{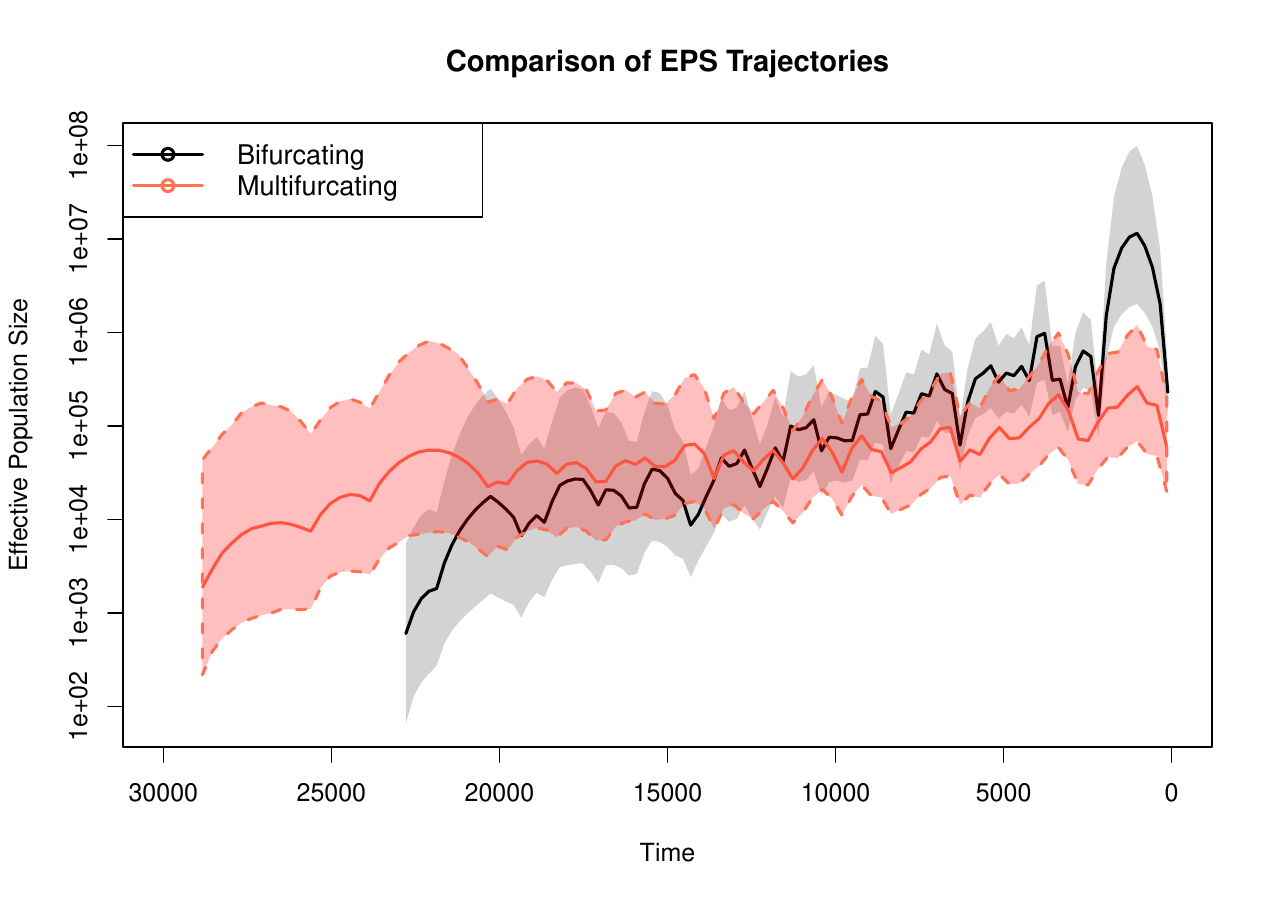}}
    \caption{\textbf{(a):} The serially sampled UPGMA binary genealogy of all Japanese sardine sequences. \textbf{(b):} Comparison of the reconstructed $N_e(t)$ using BNPR for the binary UPGMA and the MCMC $N_e(t)$ for the multifurcating tree in Figure~\ref{fig:sardines_tree}(a). Note that the $T_{MRCA}$ for the multifurcating tree in Figure~\ref{fig:sardines_tree}(a) is around 29,000 years versus 23,000 years for the binary genealogy shown in (a), which is why the inferred $N_e(t)$ profiles have different ranges for $t$.}
    \label{fig:sardines_upgma_ver}
\end{figure}

Finally, we compare the inferred $N_{e}(t)$ trajectories under a $\Lambda$-coalescent model versus the Kingman's variable $N_{e}(t)$ model. We use the methodology from \cite{drummond2000reconstructing} to reconstruct the UPGMA binary genealogy of serial samples under a fixed molecular clock with Poisson mutations. The implementation can be found in the \texttt{phylodyn} package. 
We used the BNPR function in \texttt{phylodyn} to estimate $N_{e}(t)$ using the INLA approximation \citep{pal12}. Figure~\ref{fig:sardines_upgma_ver} shows the inferred binary tree with a root date of around 23,000 years ago and the comparison of estimated $N_{e}(t)$ trajectories. The growth of $N_e(t)$ is much steeper under the bifurcating model than the multifurcating model, which suggests the inadequacy of the bifurcating model given previous work. 

\newpage 
\section{Discussion} \label{sec:discussion}

In this manuscript we propose three Bayesian nonparametric methods to jointly estimate the effective population size and $\alpha$ parameter from a multifurcating tree under the Beta$(2-\alpha, \alpha)-$coalescent model. Two of the proposed approaches iterate estimation between  $\alpha$ and $N_{e}(t)$, and the third approach approximates the posterior distribution of parameters via MCMC. The first two methods can be easily extended to any $\Lambda$-measure by computing the rates \[ \lambda_{b,k} = \int_0^1 x^{k-2} (1-x)^{b-k} \Lambda(dx)\] for each $b\geq k \geq 2$. Our implementations in \verb|phylodyn| through the function \verb|BNPR_Lambda()| allow the user to supply any discrete or continuous probability measure $\Lambda(dx)$ on $[0,1]$. As an example, in Figure~\ref{fig:Lambda_coal_BNPR_example} we show estimation of $N_{e}(t)$ of human Influenza A in New York\footnote{Genealogy is estimated from sequences analyzed in \cite{rambaut2008genomic} and available in \texttt{phylodyn}.} under four different models: Kingman's coalescent, Beta-coalescent with $\alpha=1.5$, a discrete probability mass function taking values $0.2, 0.7,0.9,1$ with equal probability, and a truncated standard normal distribution.
We choose to focus on the Beta-coalescent due to its higher applicability and nice theoretical properties. \\

\begin{figure}[h]
    \centering
    \includegraphics[width=0.9\linewidth]{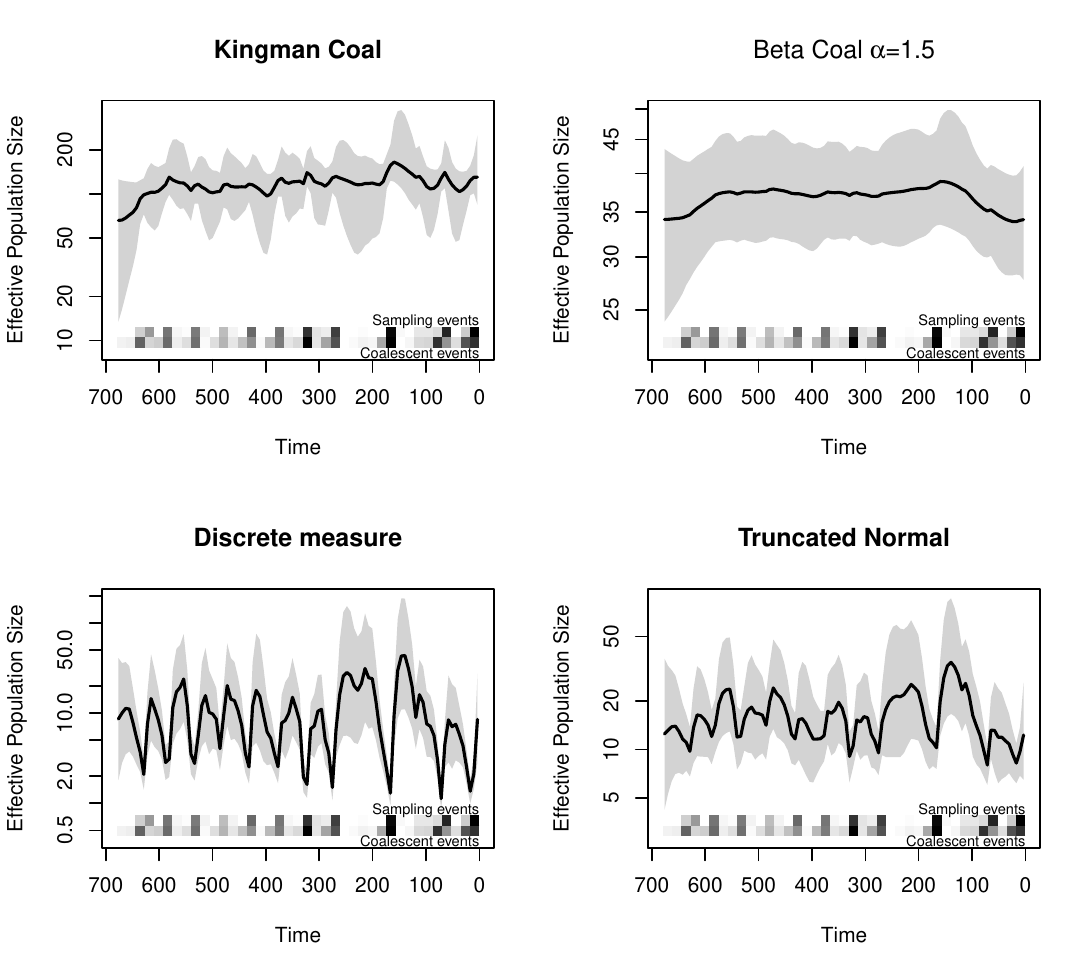}
    \caption{An example of inference of effective population size of a phylogenetic tree using BNPR under various coalescent models.}
    \label{fig:Lambda_coal_BNPR_example}
\end{figure}

This work expands on previous phylodynamic reconstruction methods developed for binary genealogies and implemented in \texttt{phylodyn} \citep{karcher2017phylodyn}. We empirically demonstrate that MCMC inference performs the best on simulated data when the underlying $N_e(t)$ is not constant. We also show that although coalescent times and block sizes are the sufficient statistics for estimating $N_{e}(t)$ and $\alpha$, $\alpha$ can be estimated by maximum likelihood from block size information alone. Surprisingly, this method proves quite accurate. In contrast, we show that $\alpha$ cannot be estimated from coalescent times alone. Although testing is out of the scope of the present manuscript, we anticipate that a likelihood ratio test for $\alpha$ would work well in this setting. \\

Finally, our methods assume that a genealogy is available without error, or rather is presented to the practitioner. We include a small simulation example in Appendix~\ref{appendix:reconstruction} comparing parameter estimates from the true simulated genealogy versus reconstructed genealogies from simulated molecular sequences. As expected, there is larger uncertainty when genealogies are estimated, in particular in the estimation of $N_e(t)$. A natural future direction is to incorporate our approach into a full Bayesian approach that targets the posterior distribution of $N_{e}(t)$, $\alpha$, and $\mathbf{g}$ from molecular sequences directly. We anticipate this to be a challenging problem as the state space of multifurcating trees is much larger than the space of binary trees.

\newpage 
\bibliographystyle{abbrvnat}
\bibliography{multifurcating}

\newpage 

\section{Appendix}

\subsection{Upper Bound on the Total Coalescent Rate}\label{appendix:c}

\begin{proposition}\label{prop:upperbound}
    The total coalescent rate $\lambda_b$ of the Beta-coalescent is upper-bounded by 
    \begin{equation} \label{eq:total_coal_rate_ub}
 \lambda_b= \sum^{b}_{k=2}\binom{b}{k}\lambda_{b,k}\leq (b-1)\left(\frac{b}{2}\right)^{\alpha-1} 
\end{equation}
\end{proposition}
\begin{proof}
We will use Wendel's Inequality \citep{qi2013bounds}: 
\begin{equation} \label{eq:wendels}
    x^{1-s}\leq \frac{\Gamma(x+1)} {\Gamma(x+s)}\leq (x+s)^{1-s}, \;\; 0<s < 1, x > 0
\end{equation}
In addition, 
\[ \Gamma(2-\alpha)\Gamma(\alpha) = (1-\alpha) \Gamma(1-\alpha) \Gamma(\alpha) = \frac{(1-\alpha) \pi }{\sin ( (1-\alpha) \pi ) } \] which implies \[ \frac{1}{\Gamma(2-\alpha)\Gamma(\alpha)} =\frac{ \sin\big ( (1-\alpha) \pi \big ) }{(1-\alpha) \pi } \in (0,1]. \]
Applying the reciprocal of Wendel's Inequality twice, we have for $2 \leq k \leq b$, 
\begin{align*}
    \binom{b}{k}\lambda_{b,k}(\alpha) & =\frac{b!\Gamma(k-\alpha)\Gamma(\alpha+b-k)\Gamma(2)}{k!(b-k)!\Gamma(2-\alpha)\Gamma(\alpha)\Gamma(b)}\\
    & =b\left(\frac{\Gamma(k-2+2-\alpha)}{k\Gamma(k)}\right)\left(\frac{\Gamma(b-k+1+(\alpha-1))(b-k+1)}{\Gamma(b-k+2)}\right)\frac{1}{\Gamma(2-\alpha)\Gamma(\alpha)} \qquad \qquad (\ast) \\
    &\leq b \times \frac{(k-2)^{1-\alpha}}{k(k-1)} \times (b-k+1)^{\alpha-1}  \\
    &\leq b\times \frac{(k-1)^{1-\alpha}}{k(k-1)}\times \left(b-(k-1)\right)^{\alpha-1} \\
    &= b\times \frac{1}{k(k-1)}\times \left(\frac{b}{k-1}-1\right)^{\alpha-1} \\
    &\leq b\times \frac{1}{k(k-1)} \times \left(\frac{b}{2}\right)^{\alpha-1}
\end{align*}
In $(\ast)$, we assumed $1<\alpha<2$. A similar proof will work for $0<\alpha<1$. Since $\sum^{b}_{k=2}\frac{1}{k(k-1)}=\frac{b-1}{b}$, we obtain the desired upper bound on the total coalescent rate when there are $b$ lineages: 
\begin{equation} \label{eq:total_coal_rate_ub}
 \lambda_b= \sum^{b}_{k=2}\binom{b}{k}\lambda_{b,k}\leq (b-1)\left(\frac{b}{2}\right)^{\alpha-1} 
\end{equation}
\end{proof}

\noindent Figure~\ref{fig:log_coal_rate_w_approx} shows that for $\alpha \in [1,2)$, the upper bound is a very good approximation. There may be some numerical imprecisions in the calculation of $\lambda_b$.

\begin{figure}[!h]
    \centering
    \includegraphics[width=0.7\linewidth]{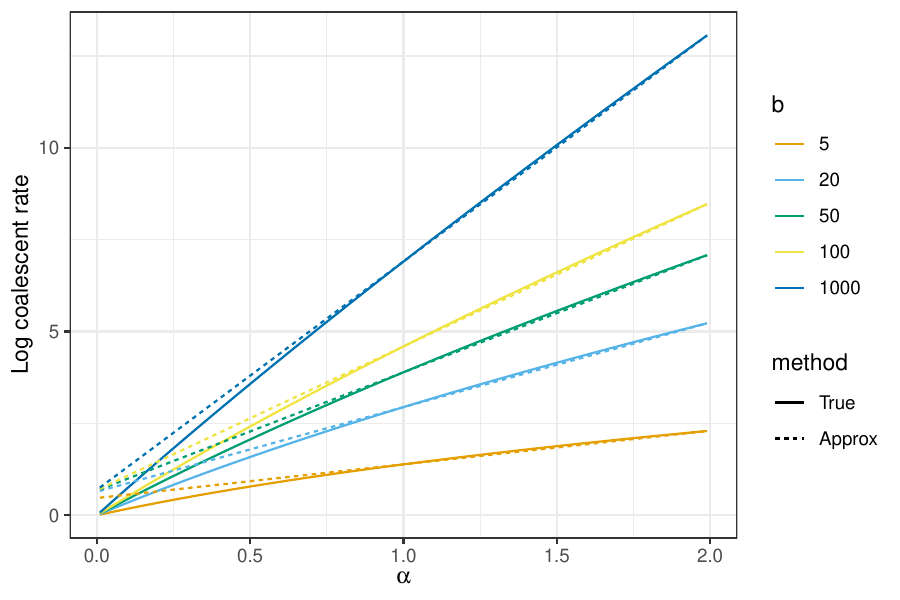}
    \caption{\textbf{Log total coalescent rate when there are $b$ lineages.} Solid lines show the true value for different values of $b$ and dashed lines show our approximation.} 
    \label{fig:log_coal_rate_w_approx}
\end{figure}

\subsection{Definition of Performance Measures}\label{appendix:a}
The performance measures for evaluating estimates of $N_e(t)$ are calculated as follows: 
\begin{align*}
    \text{Coverage} &= \frac{1}{D} \sum_{i=1}^D \mathds{1} \big ( \hat{N}^H_e(t_i) > N^T_e(t_i) > \hat{N}^L_e(t_i) \big)  \\
    \text{Bias} &= \frac{1}{D} \sum_{i=1}^D \frac{\hat{N}_e(t_i) - N^T_e(t_i)}{N^T_e(t_i)}  \\
    \text{Deviance} &= \frac{1}{D} \sum_{i=1}^D \frac{\mid \hat{N}_e(t_i) - N^T_e(t_i)\mid }{N^T_e(t_i)}  \\
    \text{MSE} &= \frac{1}{D} \sum_{i=1}^D \frac{(\hat{N}_e(t_i) - N^T_e(t_i))^2}{N^T_e(t_i)} 
\end{align*}
where $\hat{N}_e(t)$ denotes the estimated trajectory, $\hat{N}^H_e(t), \hat{N}^L_e(t)$ the upper and lower bounds of the credible regions respectively, and $N^T_e(t)$ the true trajectory, with grid points $\{t_i: i=1,..., D\}$.

\newpage 

\subsection{Results from Simulations}\label{appendix:b}
\begin{table}[H]
\centering
\begin{tabular}{|l|llll|}
  \hline
 & True $\alpha$ & BS MLE & Hybrid & MCMC \\ \hline
  Coverage & \textbf{100\%} & \textbf{100\%} & 95\% &\textbf{100\%} \\ 
  Bias & \textbf{0.003} & 0.366 & 1.027 & 0.217 \\ 
  Deviance &\textbf{0.266}& 0.49 & 1.094 & 0.349 \\ 
  MSE &  \textbf{6.491} &  40.908 & 315.304 &  12.783 \\ 
   \hline
\end{tabular}
\caption{Performance measures for estimating $N_e(t)$ of the tree in Figure~\ref{fig:sim_tree_ex}. Values were calculated by discretizing $N_e(t)$ at 100 grid points.}
\label{tab:sim_tree_ex_perf_meas}
\end{table}

\begin{figure}[H]
    \centering
    \includegraphics[width=0.6\linewidth]{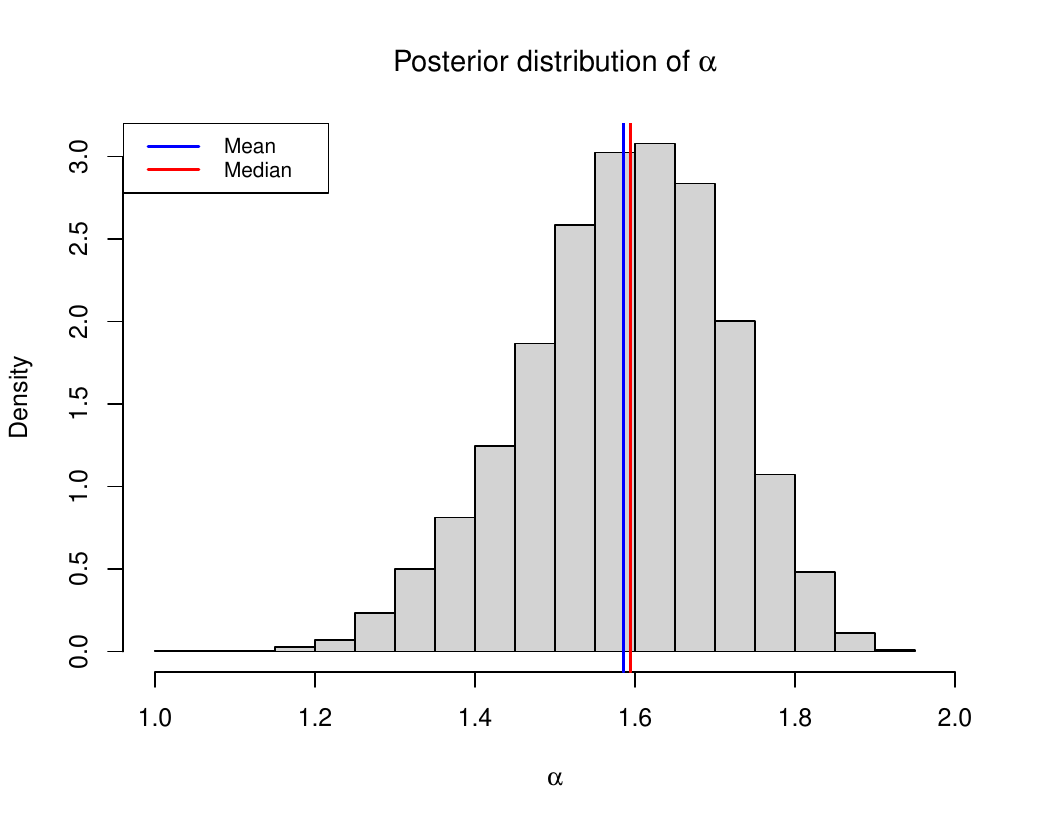}
    \caption{\textbf{Posterior distribution of $\alpha$ for the tree in Figure~\ref{fig:sim_tree_ex}}. The blue and red lines are the mean and median values respectively.} 
    \label{fig:one_tree_posterior_alpha_dist}
\end{figure}

\begin{figure}[H]
\setcounter{subfigure}{0}
    \centering
    \includegraphics[width=0.95\linewidth]{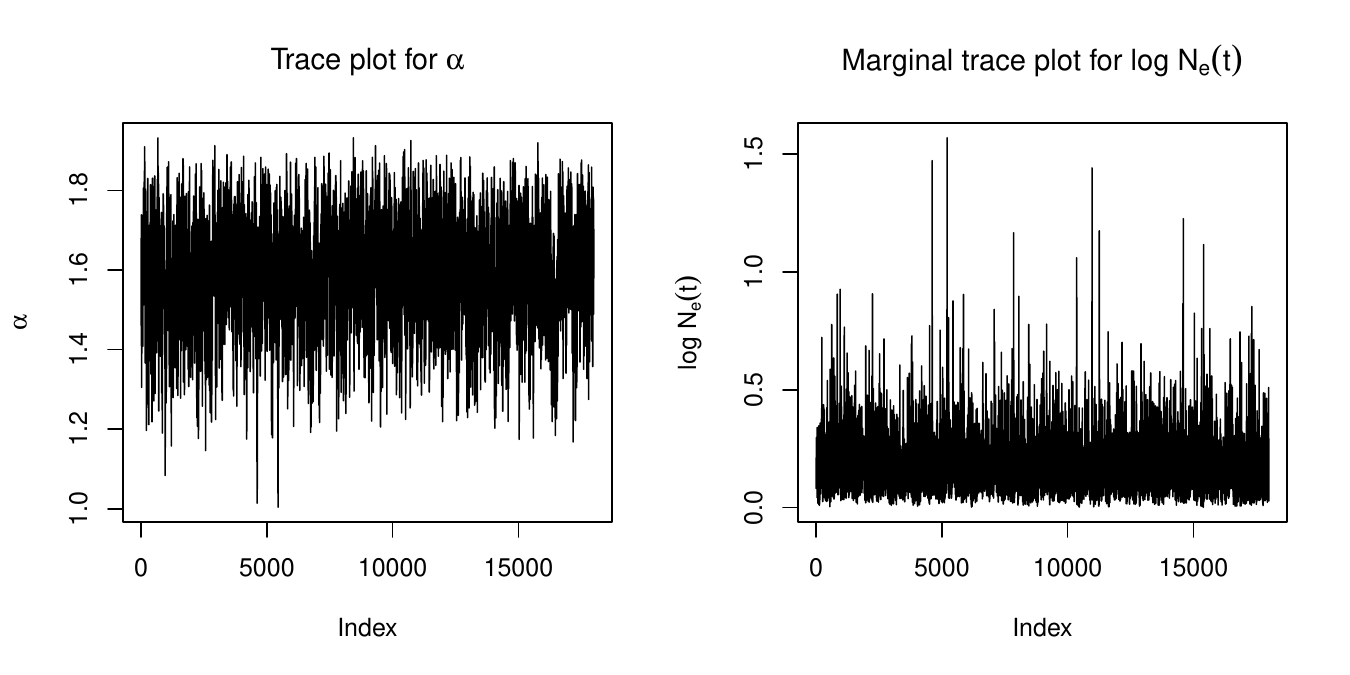}
    \caption{\textbf{Trace plots for the MCMC chain for the tree in Figure~\ref{fig:sim_tree_ex}.} The marginal trace plot for $\log N_e(t)$ is taken at $t\approx 5.23$, which is the 70th grid point.}
    \label{fig:one_tree_trace_plot}
\end{figure}

\begin{table}[H]
\centering
{\setlength{\tabcolsep}{10pt} 
\begin{tabular}{rrrrr}
  \hline
   & $n$ & Mean & Bias & MSE \\ 
  \hline
   \multirow{3}{*}{$\alpha = 0.2$ } & 50 & 0.19 & $-0.01$ & 0.02 \\ 
   & 100 & 0.18 & $-0.02$ & 0.01 \\ 
   & 500 & 0.18 & $-0.02$ & 0.01 \\ \hline 
   \multirow{3}{*}{$\alpha = 0.5$ } & 50 & 0.46 & $-0.04$ & 0.04 \\ 
   & 100 & 0.47 & $-0.03$ & 0.03 \\ 
   & 500 & 0.47 & $-0.03$ & 0.01 \\ \hline 
   \multirow{3}{*}{$\alpha = 1$ } & 50 & 0.97 & $-0.03$ & 0.09 \\ 
   & 100 & 0.98 & $-0.03$ & 0.03 \\ 
   & 500 & 0.99 & $-0.01$ & 0.01 \\ \hline 
   \multirow{3}{*}{$\alpha = 1.5$ } & 50 & 1.48 & $-0.02$ & 0.03 \\ 
   & 100 & 1.5 & $-0.008$ & 0.01 \\ 
   & 500 & 1.5 & $-0.003$ & 0.002 \\ \hline 
   \multirow{3}{*}{$\alpha = 1.8$ }& 50 & 1.8 & $-0.005$ & 0.01 \\ 
  & 100 & 1.8 & 0.001 & 0.001 \\ 
  & 500 & 1.8 & $-0.002$ & 0.001 \\ 
   \hline
\end{tabular} }
\caption{\textbf{Estimation of $\alpha$ from tree topology only.}  Mean $\hat{\alpha}^{BS}$, mean bias and MSE based on 1000 simulations for each combination of $\alpha$ and $n$.}
\label{tab:topology_only_mle}
\end{table}

\begin{table}[H]
\centering
\begin{tabular}{r|rrr}
  \hline
 Traj. & BS MLE & Hybrid & MCMC \\ 
  \hline
Unif & 26.6\% & 32.3\% &\textbf{ 41.1\%} \\ 
  Exp & 30.6\% & 28.5\% & \textbf{40.9\%}\\ 
  BB & 32.7\% & 26.6\% & \textbf{40.7\%} \\ 
   \hline
\end{tabular}
\caption{Percentage of simulations where each method exhibits the lowest deviance for estimating $\alpha$, where we simulated across different $n$ and $\alpha$ values, sampling schedules, and effective population size trajectories.}
\label{tab:best_method_alpha}
\end{table}

\begin{figure}[H]
    \centering
    \sidesubfloat[]{\includegraphics[width=0.9\linewidth]{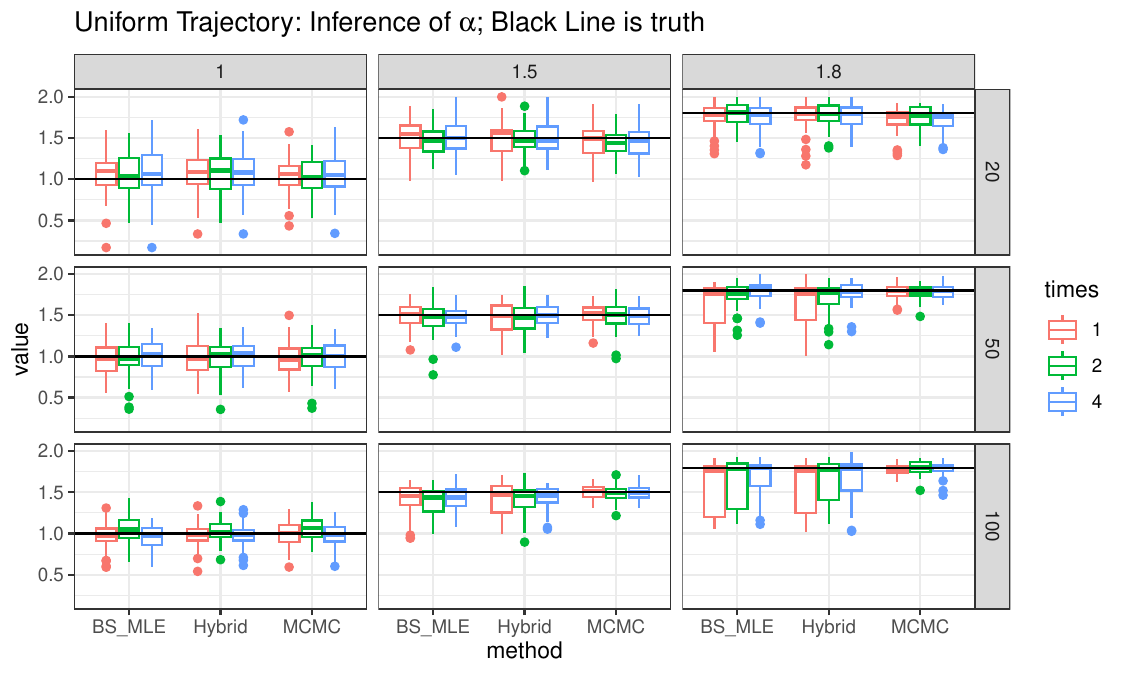}}
    \phantomcaption
\end{figure}

\begin{figure}[H] \ContinuedFloat
    \centering
    \sidesubfloat[]{\includegraphics[width=0.9\linewidth]{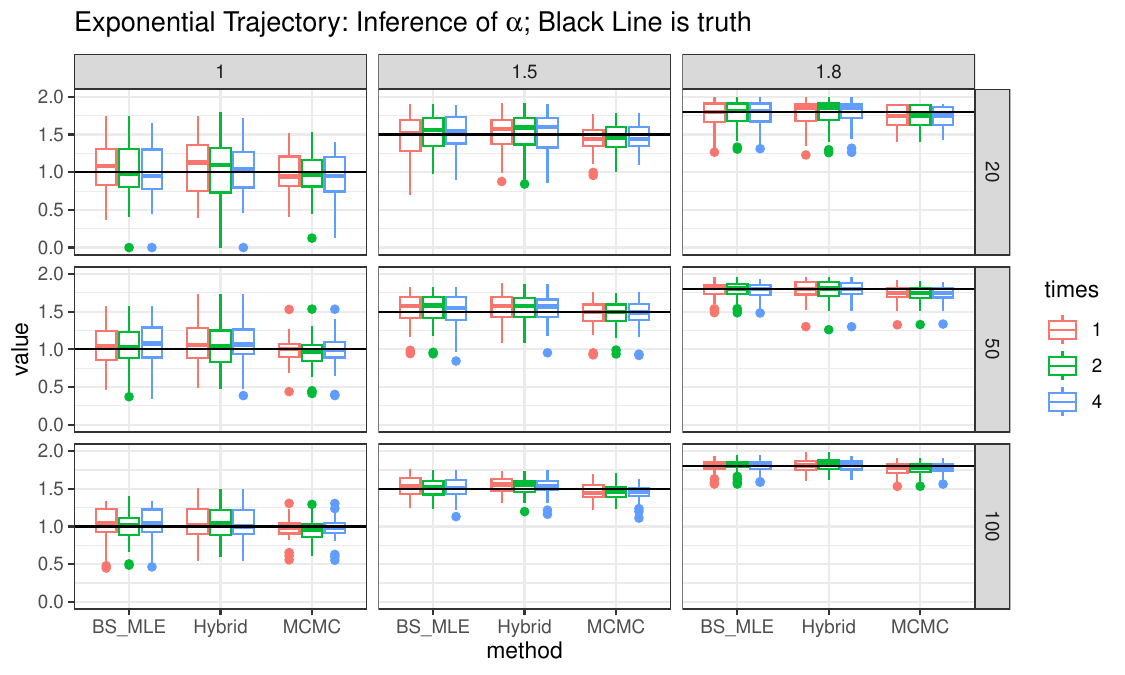}}
    \phantomcaption
\end{figure}

\begin{figure}[H] \ContinuedFloat
    \centering
    \sidesubfloat[]{\includegraphics[width=0.9\linewidth]{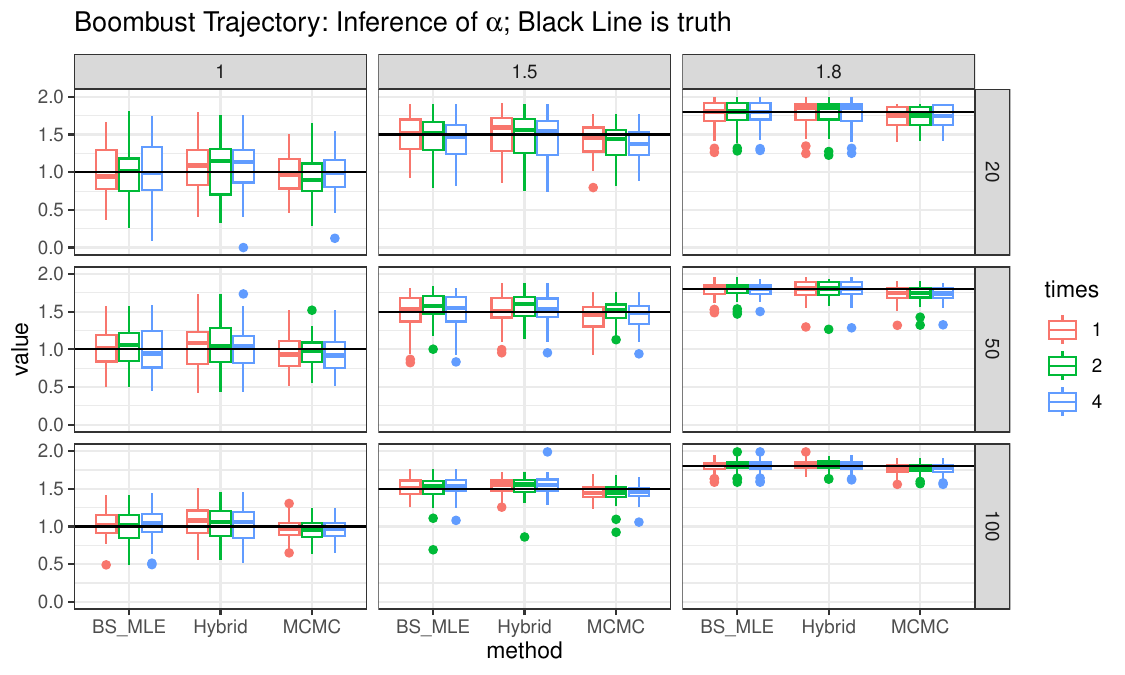}}    
    \caption{Boxplots of inferred $\alpha$ values from simulations with trees generated from (a): Uniform trajectory, (b): Exponential trajectory, (c): Boombust trajectory. Each row is a fixed $N$; each column is a fixed $\alpha$; and each colors is a sampling time. The black line is the true $\alpha$.}
    \label{fig:all_traj_alpha_boxplots}
\end{figure}

\subsection{Simulation study on reconstructed phylogenies}\label{appendix:reconstruction}

Our methods assume a multifurcating genealogy is given without error. In practice, genealogies are estimated with uncertainty from observed molecular data. To analyze the sensitivity of our estimates to genealogical estimation uncertainty, we generated three multifurcating trees with 50 tips under a Beta$(2-\alpha,\alpha)$ coalescent with $\alpha=1.5$ and an exponential effective population size trajectory. The tips were heterochronously sampled at $t=0,1,2$ with a $(40\%, 40\%, 20\%)$ split. For each of the trees shown in Figure~\ref{fig:reconstructed_original_trees}, we simulated 10 datasets of 50 sequences with 1000 base pairs using AliSim under the Jukes-Cantor mutation model \citep{ly2023alisim}. Then, we used IQ-Tree and TreeTime to reconstruct dated, multifurcating trees via maximum likelihood \citep{nguyen2015iq, sagulenko2018treetime} and applied our three methods to estimate $\alpha$ and $N_e(t)$ on these reconstructed genealogies. We also carried out the same analysis pipeline on datasets simulated with a 10 times faster mutation rate. Results are presented in Figure~\ref{fig:reconstructed_results} below. As expected, estimates of $\alpha$ from inferred genealogies are noisier than those obtained from the true genealogies, however inference of $N_e(t)$ is generally much worse when genealogies are inferred. We conclude that our methods, like any  method that assume a fixed estimated genealogy, are not robust to genealogical misestimations. A joint method for estimating the genealogy and model parameters is still an open problem. \\ 

\begin{figure}[H]
    \centering
    \includegraphics[width=0.95\linewidth]{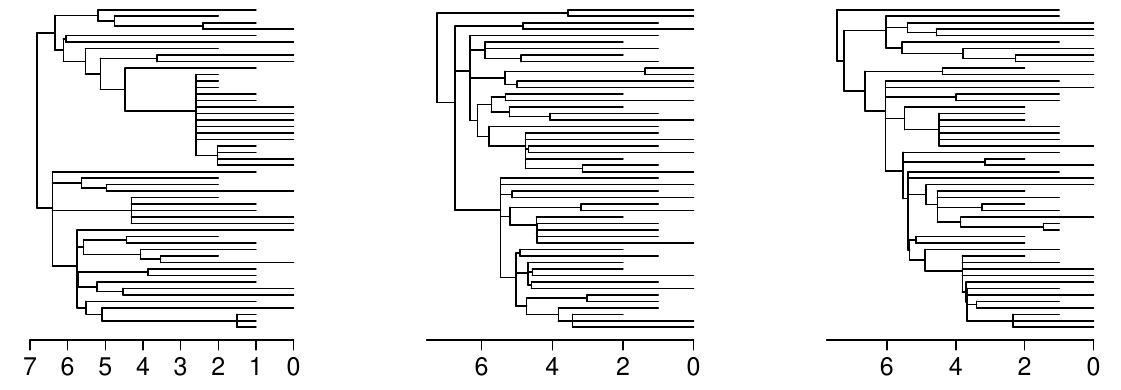}
    \caption{The three ``true'' genealogies used for simulation. Each was generated under a Beta-coalescent with $\alpha=1.5$ and an exponential effective population size trajectory.}
    \label{fig:reconstructed_original_trees}
\end{figure}

\begin{figure}[H]
    \centering
    \includegraphics[width=0.95\linewidth]{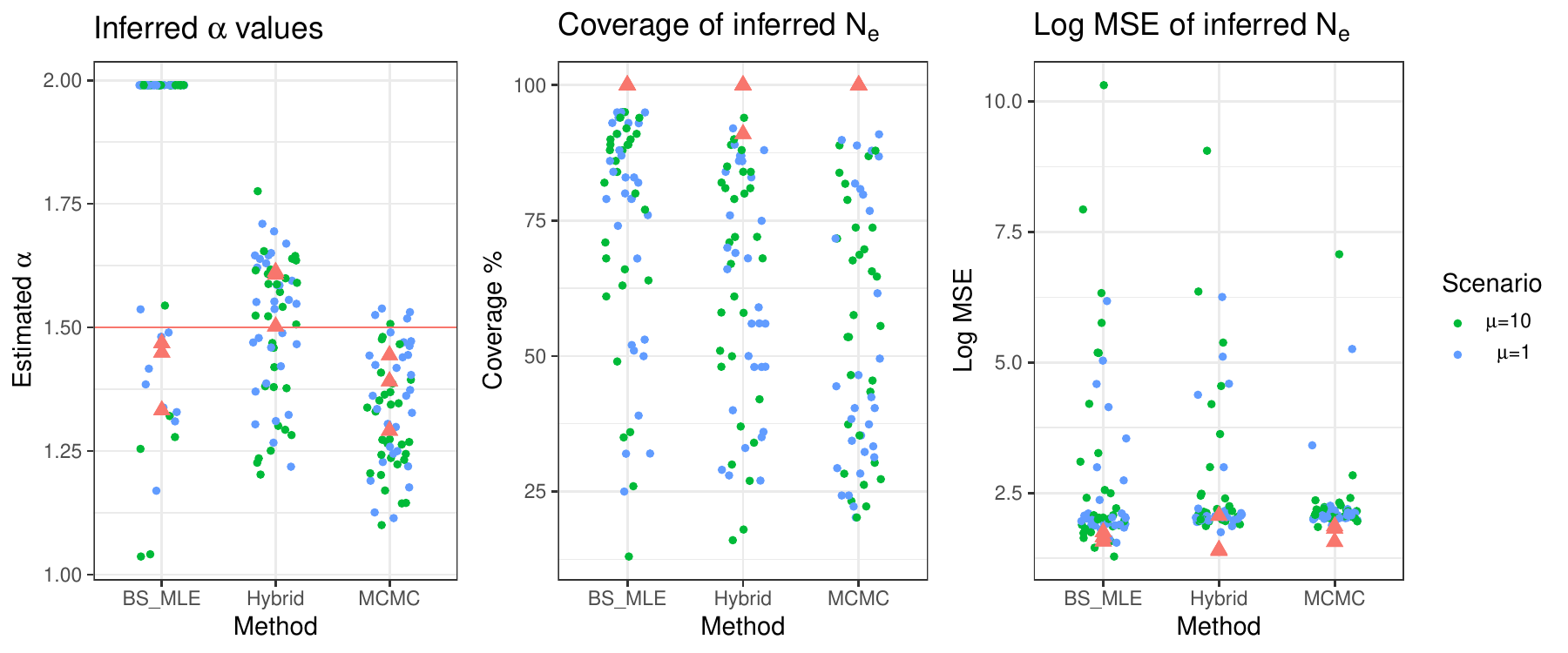}
    \caption{\textbf{Results of inferred $\alpha$ and $N_e(t)$ values under the block-size MLE, hybrid, and MCMC methods based on reconstructed genealogies.} (a) Scatterplots of inferred $\alpha$ values: the horizontal line is the true $\alpha=1.5$. (b) Scatterplots of the coverage of the 95\% credible intervals of the inferred $N_e(t)$ values. (c) Scatterplots of the $\log_{10}$ MSE of the inferred $N_e(t)$ values. Note the red triangles in each figure represent the value from applying our methods on the true genealogy in Figure~\ref{fig:reconstructed_original_trees}. The green and blue colors represent the fast mutation scenario, and the original mutation scenario respectively.}
    \label{fig:reconstructed_results}
\end{figure}

\end{document}